\documentclass{scrartcl}

\usepackage{algorithm}
\usepackage[noEnd,indLines]{algpseudocodex}
\usepackage{amsmath}
\usepackage{amssymb}
\usepackage{amsthm}
\usepackage{authblk}
\usepackage[shortlabels]{enumitem}
\usepackage{framed}
\usepackage{ifthen}
\usepackage{mathtools}
\usepackage{thmtools}

\usepackage[colorlinks]{hyperref}
\usepackage{cleveref}
\usepackage{mathtools}

\hypersetup{linkcolor=blue,citecolor=blue}

\newtheorem{theorem}{Theorem}[section]
\newtheorem{lemma}[theorem]{Lemma}
\newtheorem{corollary}[theorem]{Corollary}
\newtheorem{observation}[theorem]{Observation}
\newtheorem{proposition}[theorem]{Proposition}

\theoremstyle{definition}
\newtheorem{definition}[theorem]{Definition}

\newcommand{\N}{\mathbb{N}}

\newcommand{\alg}[1]{\textup{\texttt{#1}}}
\newcommand{\problem}[1]{\textup{\textsc{#1}}}
\newcommand{\Time}{\mathbb{T}}
\newcommand{\pDTM}{\problem{Decremental Tree Minima}}
\newcommand{\pDTR}{\problem{Decremental Tree Roots}}
\newcommand{\pCT}{\problem{Cartesian Tree}}
\newcommand{\pPM}{\problem{Path Minima}}
\newcommand{\pTS}{\problem{Tree Sorting}}
\newcommand{\pTTS}{\problem{Top-$k$ Tree Sorting}}

\newcommand{\Ind}{\mathcal{C}} %
\newcommand{\ETs}{\mathcal{E}} %
\newcommand{\Lin}{\mathcal{N}} %
\newcommand{\IndE}{\Ind^{\mathrm{E}}} %
\newcommand{\EPTs}{\mathcal{E}^{\mathrm{E}}} %
\newcommand{\LEs}{\mathcal{L}} %

\newcommand{\wmin}{w^{\mathrm{min}}}

\newcommand{\fO}{\mathcal{O}}
\newcommand{\fP}{\mathcal{P}}

\newcommand{\Prios}{\mathfrak{P}} %

\newcommand{\gminus}{-} %

\algrenewcommand\textproc{\normalfont\texttt} %
\makeatletter
\algnewcommand\algorithmicforeach{\textbf{for each}}
\algdef{S}[FOR]{ForEach}[1]{%
	\algpx@startCodeCommand\algpx@startIndent\algorithmicforeach\ #1\ \algorithmicdo%
}
\pretocmd{\ForEach}{\algpx@endCodeCommand}{}{}
\makeatother

\title{Universally Optimal\\Decremental Tree Minima}
\author{Benjamin Aram Berendsohn}
\affil{\normalsize{}Max Planck Institute for Informatics}
\date{\vspace{-5ex}}

\begin{document}
	
\maketitle

\begin{abstract}
An algorithm on weighted graphs is called \emph{universally optimal} if it is optimal for \emph{every} input graph, in the worst case taken over all weight assignments.
Informally, this means the algorithm is competitive even with algorithms that are optimized for only one specific input graph.
Universal optimality was recently introduced [Haeupler et al.\ 2024] as an alternative to the stronger, but often unachievable \emph{instance optimality}.

In this paper, we extend the concept of universal optimality to data structures.
In particular, we investigate the following dynamic graph problem:
Given a vertex-weighted forest, maintain the minimum-weight vertex of every tree under edge deletions.
The problem requires $\Theta(\log n)$ amortized time per operation in general, but only $\mathcal{O}(1)$ time if the initial forest is a path.

We present a data structure that has optimal total running time for every fixed initial forest and every fixed \emph{number} of operations/queries $m$, when taking the worst case over all weight assignments and operation sequences of length~$m$. This definition of universal optimality is easily adapted to other data structure problems.

Our result combines two techniques: (1) A decomposition of the input into paths, to take advantage of the $\mathcal{O}(1)$-time path-specific data structure; and (2) splay trees [Sleator and Tarjan 1985], which, informally speaking, are used to optimally handle a certain sorting-related subproblem.
We apply our data structure to solve problems related to \emph{Cartesian trees}, \emph{path minimum queries}, and \emph{bottleneck vertex/edge queries}, each with a certain universal optimality guarantee. Our data structure also can be modified to support \emph{edge} weights instead of vertex weights. Further, it generalizes to support \emph{semigroup sum} queries instead of \emph{minimum} queries, in universally optimal time.
\end{abstract}

\pagebreak
\tableofcontents
\pagebreak

\section{Introduction}

In this paper, we study the following dynamic graph problem, which we call \pDTM{} (DTM). We are given a vertex-weighted forest, and need to support the following two operations: $\alg{delete-edge}(u,v)$, which does the obvious; and $\alg{tree-min}(v)$, which returns the minimum vertex in the connected component (i.e., the tree) that contains $v$. Note that vertex weights cannot be changed and edges cannot be inserted. We assume weights are pairwise distinct and can only be accessed via comparisons.

To the author's knowledge, previous works only considered the \emph{fully-dynamic} variant of this problem, where edge insertions are also allowed~\cite{GoldbergGrigoriadisEtAl1991}.
Standard dynamic forest data structures like \emph{link-cut trees}~\cite{SleatorTarjan1983,SleatorTarjan1985} %
solve the fully-dynamic problem in $\fO(n \log n)$ time, when there are $n$ vertices and $\Theta(n)$ operations.
This is tight by a reduction from dynamic connectivity~\cite{PatrascuDemaine2006}, even when the forest is a collection of paths at all times.

For the \pDTM{} problem (where insertions are not allowed), the same $\Omega(n \log n)$ lower bound holds by a reduction from sorting, so the above solutions are worst-case optimal.
What makes the DTM problem interesting is that it is \emph{easy} when the input tree is a path: Using a standard \emph{range minimum query} data structure~\cite{BenderFarach-Colton2000}, we can achieve $\fO(m+n)$ time for $m$ operations and $n$ vertices.

This begs the question whether we can beat $\Theta(\log n)$ amortized time per operation on other input forests. In this paper, we give a complete answer to this question: A~data structure that has optimal running time on \emph{every input forest}.
To formalize this, we extend the notion of \emph{universal optimality}, recently defined for (non-dynamic) graph algorithms~\cite{HaeuplerHlad'ikEtAl2024}, to data structures.
We also give several applications of our data structure, which yields universally optimal \emph{algorithms}.
In the following, we first discuss universal optimality in more detail, and then present our results.

\paragraph{Universal optimality for algorithms.}

In classical worst-case complexity, the running time of an algorithm is usually measured in terms of input size only; i.e., for a fixed input size, take the worst case over all inputs of that size. For graph algorithms, sometimes both the number $n$ of vertices and the number $m$ of edges is fixed, and parametrized graph algorithms additionally fix the value of some parameter.
In universal optimality, the parameter is the \emph{graph itself}, and the worst case is only taken over additional data like vertex weights.

More precisely, consider a problem where the input is a graph $G$ with an (edge- or vertex-)weight function $w$, and let $A$ be an algorithm for this problem. Let $\Time(A,G,w)$ denote the running time of $A$ on input $(G,w)$, and let $\Time(A,G) = \max_w \Time(A,G,w)$, where the maximum is taken over all possible weight functions for $G$.\footnote{We assume here that running times are integers and that the maximum exists.}
Then $A$ is \emph{universally optimal} if, for every graph $G$, and every other algorithm $A'$, we have $\Time(A,G) \le c \cdot \Time(A', G)$ for some globally fixed constant $c$.
In particular, universal optimality means that $A$ is competitive with algorithms that are specifically engineered for a single graph.
In contrast, \emph{instance optimality} means that we have $\Time(A,G,w) \le c \cdot \Time(A', G,w)$ for \emph{every input} $(G,w)$ and some globally fixed constant $c$. Instance optimality can be seen as the ``ultimate'' optimality guarantee;\footnote{If we ignore constant factors.} unfortunately, it is often unachievable. This is in particular true for problems that contain sorting as a special case. Indeed, every sequence can be sorted in linear time by \emph{some} algorithm, but for every sorting algorithm there exists a sequence that requires $\Omega(n \log n)$ time.

We now discuss some problems where universally optimal algorithms are known.
First, in the \emph{distance ordering problem}, we are given a directed edge-weighted graph $G$ and a starting vertex $s$, and need to order the vertices by distance from $s$.
Haeupler, Hladík, Rozhon, Tarjan, and Tětek~\cite{HaeuplerHlad'ikEtAl2024} showed that Dijkstra's algorithm, equipped with a special priority queue, solves this problem in universally optimal running time.

A related problem is \emph{partial order sorting}. Here, we are given a set of elements and the Hasse diagram of a partial order on these elements. We need to sort the elements using comparisons, under the promise that the sorted order is a linear extension of the given partial order. For universal optimality, the Hasse diagram corresponds to the graph~$G$, and the underlying order takes the role of the weight assignment.
Two different universally optimal algorithms for this problem have been found recently~\cite{HaeuplerHladikEtAl2025a,vanderHoogRotenbergEtAl2025}.

Finally, there is an algorithm for \emph{convex hull}~\cite{KirkpatrickSeidel1986,AfshaniBarbayEtAl2017} that has been called universally optimal~\cite{HoogRotenbergEtAl2025}. This is obviously not a graph problem; however, the concept of universal optimality can be generalized to arbitrary input types, as long as we can identify a \emph{primary input part} (replacing the graph), and a \emph{secondary input part} (replacing the weight assignment). For convex hull, the primary input part is taken as the \emph{set} of input points, and the secondary input part is the \emph{order} in which they are given.
A universally optimal algorithm by this definition is \emph{instance-optimal} among all algorithms that behave the same irrespective of input order; accordingly, the term used originally was \emph{instance optimality in the order-oblivious model}~\cite{AfshaniBarbayEtAl2017}.

Note that all three problems have worst-case lower bounds of $\Omega(n \log n)$ (where $n$ denotes the number of vertices in the graph, the number of elements to sort, or the number of points). These lower bounds are achieved if the input is a star, a partial order consisting of a single antichain, or a set of points in convex position, respectively. In all three cases, worst-case optimal algorithms have been known for a long time.

On the other hand, the following easy cases can be identified: For the distance-ordering problem, if $G$ is a directed path starting at $s$, then the problem can be solved in linear time, regardless of the edge weights. Similarly, if the given partial order is a chain, then no comparisons are necessary to determine the sorted order. Finally, it is known that a convex hull can be computed in linear time if it only contains a constant number of points, regardless of the input order~\cite{KirkpatrickSeidel1986,Chan1996}.

Thus, worst-case running times range between $\Theta(n \log n)$ and $\Theta(n)$ depending on the primary input part (i.e., the graph, partial order, or point set). It is reasonable to ask for some \emph{interpolation} between these extremes.
A well-known result in that direction are the output-sensitive convex hull algorithms of Kirkpatrick--Seidel~\cite{KirkpatrickSeidel1986} and Chan~\cite{Chan1996}, which take $\fO(n \log h)$ time, where $h$ is the number of points on the convex hull.
Universally optimal algorithms can be seen as a \emph{best possible} such interpolation.

\paragraph{Universal optimality for data structures.}
Recall the \pDTM{} (DTM) problem defined above. We start with a vertex-weighted forest, we can delete an edge with $\alg{delete-edge}(u,v)$, and we can find the minimum-weight vertex in the tree containing a given vertex $v$ with $\alg{tree-min}(v)$.
In the following, we call vertex weights \emph{priorities}, to distinguish them from numerical weights. We assume priorities in the input forest are pairwise distinct and can only be accessed via comparisons.

When the input forest $G$ has $n$ vertices, then $\Theta(n)$ DTM operations can be performed in $\fO(n \log n)$ time~\cite{GoldbergGrigoriadisEtAl1991}.
This is tight if $G$ is a star where the central vertex has priority $\infty$. Indeed, in this case a DTM data structure can \emph{sort} the leaf priorities with alternating \alg{tree-min} queries and edge deletions, which requires $\Omega(n \log n)$. On the other hand, if $G$ is path, then DTM can be solved with $\fO(n)$ time preprocessing and $\fO(1)$ time per operation (details are given in \cref{sec:pre:pm}).

We are thus in a similar situation as for the aforementioned algorithmic problems: The worst-case running time for $\Theta(n)$ operations a fixed graph can be $\Theta(n \log n)$ or $\Theta(n)$, depending on the graph (in fact, it can be anything in between).
However, it is not immediately obvious how to define universal optimality for data structures.
In particular, the restriction to $\Theta(n)$ operations may not adequately capture the complexity of the problem.
On the other hand, if the number of operation is arbitrary, the worst-case total running time for any fixed graph is unbounded.

We hence use the following definition, which can be easily adapted to a wide range of dynamic graph problems (see below). We model a problem instance as a tuple $(G,p,\sigma)$, where $G$ is the initial tree, $p$ is the priority assignment and $\sigma$ is the sequence of operations.
The primary input part for universal optimality is $(G,|\sigma|)$, i.e., the graph and the \emph{number} of operations.
Note that $\sigma$ is revealed one-by-one, and not even $|\sigma|$ is known at the start.

More precisely, if $D$ is a data structure for the DTM problem, $G$ is a graph, and $m \in \N$, then $\Time(D,G,m)$ denotes the total time $D$ requires for an input $(G,p,\sigma)$, in the worst case over all possible priority assignments $p$ and all valid operation sequences $\sigma$ with $|\sigma| = m$. We call $D$ universally optimal when, for each $G$, each $m$, and each other data structure $D'$, we have $\Time(D,G,m) \le c \cdot \Time(D',G,m)$ for some globally fixed constant~$c$.
Our main result %
is:

\begin{theorem}\label{p:dtm-main}
	There is a universally optimal data structure for the \pDTM{} problem.
\end{theorem}

\paragraph{Alternative definitions.} Our definition is not the only reasonable interpretation of universal optimality. For the DTM problem specifically, one could require \alg{tree-min} to run in constant time, and measure the total time for preprocessing plus the at most $n$ edge deletions. This is in line with the idea of ``maintaining'' tree minima. However, our definition is stronger and more directly generalizes to other decremental graph problems.

Another definition may count each operation separately. For example, for DTM, we may consider pairs $m_1, m_2 \in \N$ instead of a single value $m$, and take the worst case over all sequences of $m_1$ edge deletions and $m_2$ \alg{tree-min} queries. 
Our data structure does not achieve this definition of universal optimality, but perhaps it could be adapted with some mechanism that postpones edge deletions until required.

For an even stronger definition, we could fix the whole sequence of operations, taking the worst case over only the input weights. An algorithm satisfying this definition would have to match the running time of an optimal \emph{offline} algorithm, which is impossible for DTM, as is explained at the end of this section.
However, for a different problem where operations have weight parameters, the following definition could be interesting: Fix the operations and their vertex parameters, but not the weight parameters.\footnote{This definition was suggested by an anonymous reviewer.}

For \emph{incremental} graph data structures, we can easily adapt our definition by taking the final graph instead of the initial graph. For \emph{fully-dynamic} problems, one could take the graph containing all edges ever present in the graph.\footnote{This may seem crude, but there are applications where this graph is meaningful. For example, a maximum flow algorithm of Goldberg and Tarjan~\cite{GoldbergTarjan1988} uses a dynamic forest data structure, and all edges ever inserted into the dynamic forest are contained in the input graph.}

\paragraph{Cartesian trees on graphs.}

\begin{figure}
	\centering
	\begin{minipage}{.2\textwidth}
		\centering
		\includegraphics[scale=1.5]{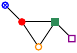}
	\end{minipage}%
	\hspace{5mm}%
	\begin{minipage}{.2\textwidth}
		\centering
		\includegraphics[scale=1.5]{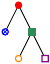}
	\end{minipage}%
	\hspace{5mm}%
	\begin{minipage}{.4\textwidth}
		\centering
		\includegraphics[scale=1.5]{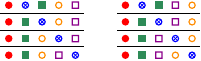}
	\end{minipage}
	\caption{A graph (right), an elimination tree on that graph (center), and the eight priority-orderings of vertices inducing that elimination tree (right).}\label{fig:et-ind}
\end{figure}

Our second result %
is a universally optimal algorithm for computing a certain graph decomposition, using our \pDTM{} data structure. We start with some definitions.
An \emph{elimination tree} on a connected graph $G$ is a rooted tree $T$ with $V(T) = V(G)$, constructed as follows (see \cref{fig:et-ind}, left/center): First, pick an arbitrary vertex $r \in V(G)$ as the root of $T$. Then, recursively construct elimination trees on each component of $G \gminus r$, and add them as child subtrees to $r$. Note that if $G$ is a path, then $T$ is a binary tree, and fixing a direction on the path induces the usual left and right children. Elimination trees have been studied in the context of matrix factorization, vertex search, and structural graph theory; see \cref{sec:related} for more details.

A priority assignment $p$ uniquely determines an elimination tree $T$ on $G$ as follows (see \cref{fig:et-ind}): Whenever a subtree root is picked, pick the currently available vertex with minimum priority. Alternatively, $T$ is the unique elimination tree on $G$ such that for each ancestor-descendant pair $u, v$, we have $p(u) < p(v)$. We call $T$ the \emph{Cartesian tree on $(G,p)$} (since it generalizes the concept of so-called Cartesian trees on sequences~\cite{Vuillemin1980}), and write $\Ind(G,p) = T$. The \pCT{} problem consists of computing $\Ind(G,p)$, given $G$ and $p$. Some variants of the \pCT{} problem have been considered before~\cite{Liu1986,Chazelle1987,ZhuMutchler1994,DeanMohan2013,DemaineLandauEtAl2014}.

In this paper, we give an algorithm for \pCT{}. Our algorithm runs in time $\fO( \log |\ETs(G)| )$, where $\ETs(G)$ denotes the set of elimination trees on $G$. Since every elimination tree on $G$ is the Cartesian tree for some priority function, any algorithm needs $\log |\ETs(G)|$ comparisons by the information-theoretic lower bound. Thus, our algorithm is universally optimal (with primary input part $G$).

If $G$ is a tree, then we can use our \pDTM{} data structure. Simply determine the minimum-priority vertex $v$, and remove it by deleting all adjacent edges. The resulting graph contains a component for each deleted edge. Thus, we can make $v$ the root and recurse on each component.
We can show that this algorithm runs in $\fO(\log |\ETs(G)|)$ time using techniques developed for \pDTM{}.

The general case, where $G$ is an arbitrary connected graph, reduces to the tree case by computing a certain spanning tree~\cite{ZhuMutchler1994}. We can compute this spanning tree sufficiently fast with a variant of the textbook Dijkstra-Jarník-Prim algorithm, using the analysis of the aforementioned universally optimal variant of Dijkstra's algorithm~\cite{HaeuplerHlad'ikEtAl2024}.

\begin{restatable}{theorem}{restateCartTree}\label{p:univ-opt-et}
	There is an algorithm that, given a graph $G$ and a priority assignment~$p$, computes $\Ind(G,p)$ in universally optimal time $\fO( |E(G)| + \log |\ETs(G)| )$.
\end{restatable}

Note that, crucially, we assume vertex priorities are unsorted. If they are pre-sorted, we can compute $\Ind(G,p)$ in linear time~\cite{DemaineLandauEtAl2014}.\footnote{The reference~\cite{DemaineLandauEtAl2014} discusses a variant of the problem for trees with \emph{edge} priorities, but the idea is easily adapted to vertex priorities, and generalized to graphs using the spanning tree technique.}

Cartesian trees on paths~\cite{Vuillemin1980} are useful to implement \emph{range minimum queries}~\cite{GabowBentleyEtAl1984}. Our Cartesian trees on graphs can be used similarly, as we discuss next.

\paragraph{Path minimum and bottleneck queries.}

Suppose we are given a tree $G$ with priority assignment $p$, and we want to preprocess it to answer the following query:
\begin{itemize}
	\item $\alg{path-min}(u,v) \rightarrow w$ returns the minimum-priority vertex $w$ on the path between $u$ and $v$.
\end{itemize}

Note that since $G$ is a tree, the path between $u$ and $v$ is unique. It is not hard to see that if $T$ is the Cartesian tree $\Ind(G,p)$, then the answer to $\alg{path-min}(u,v)$ is the lowest common ancestor (LCA) of $u$ and $v$ in $T$. Since we can preprocess $T$ in $\fO(n)$ time to answer LCA queries in $\fO(1)$ time~\cite{HarelTarjan1984}, we have:

\begin{restatable}{theorem}{restatePMDS}\label{p:pmds}
	There is a data structure that, given a tree $G$ and a priority assignment~$p$, can answer \alg{path-min} queries in $\fO(1)$ time, with $\fO( \log |\ETs(G)| )$ preprocessing time.
\end{restatable}

This data structure is not universally optimal in any strong sense. For example, if $G$ is a star, we can compute path minimum queries in constant time, without any preprocessing. On the other hand, we have $\log |\ETs(G)| \in \Theta(n \log n)$ if $G$ is a star. In fact, it is known that the problem can be solved with less than $\fO(n \log^* n)$ preprocessing time and constant query time~\cite{AlonSchieber1987}.

However, note that our data structure does not need to make any comparisons after the preprocessing. We can show universal optimality under this restriction:

\begin{restatable}{theorem}{restatePMDSUOpt}
	The data structure of \cref{p:pmds} has universally optimal preprocessing time if we disallow comparisons in \alg{path-min} queries.
\end{restatable}

A generalization of \alg{path-min} queries are \emph{bottleneck vertex queries}~\cite{ShapiraYusterEtAl2011}. Here, the underlying graph $G$ is not necessarily a tree. The \emph{bottleneck} of a path is the minimum priority on that path, and the query $\alg{bottleneck}(u,v)$ returns the maximum bottleneck among all paths between $u$ and $v$ in $G$.
Our data structure can also handle $\alg{bottleneck}$ queries via a simple reduction~\cite{ShapiraYusterEtAl2011}.

\paragraph{Edge-weighted graphs.}

All of our results generalize to a setting where priorities are given for \emph{edges} instead of vertices. The \problem{Decremental Tree Minima} problem and the static \alg{path-min} data structure problem generalize in the obvious way, and bottleneck vertices become bottleneck \emph{edges}, another well-studied concept~\cite{VassilevskaWilliamsEtAl2007APBE,DemaineLandauEtAl2014}.

Elimination trees also can be adapted as follows: An \emph{edge-partition tree} (EPT) on a connected graph $G$ is a rooted binary tree $T$ with $V(T) = E(G) \cup V(G)$, where edges of $G$ correspond to inner nodes of $T$, and vertices of $G$ correspond to leaves of $T$. An EPT is constructed as follows: Pick an arbitrary edge $e$ as the root of $T$, then recursively construct up to two child subtrees on the up to two components of $G \gminus e$.

An edge-priority assignment $p$ again uniquely determines an EPT on the given graph~$G$. These \emph{Cartesian EPTs} have been studied before under several different names~\cite{Chazelle1987,Neto1999,BoseMaheshwariEtAl2004,BrodalDavoodiEtAl2011,DeanMohan2013,DemaineLandauEtAl2014}.

All these edge-weighted variants reduce to the corresponding vertex-weighted problem with the following reduction: Place a vertex ``on'' each edge, i.e., replace each edge with a path of length two. That vertex receives the priority of the original edge, and all original vertices are assigned priority $\infty$. We can then apply our solutions to the resulting graph, with only minor modifications.

Since the number of newly introduced vertices is linear, worst-case optimality is maintained. We can also prove our universal optimality guarantees with some additional work.
In particular, our universally optimal algorithm for computing Cartesian EPTs strictly improves upon early $\fO(n \log n)$ solutions for trees~\cite{Chazelle1987}, as well as the $\fO(n \log \ell)$ algorithm of Dean and Mohan for trees with $\ell$ leaves~\cite{DeanMohan2013}.

\newcommand{\pDSTS}{\problem{Decremental Semigroup Tree Sums}}

\paragraph{Semigroups.} Consider the following variation of the DTM problem, which we call \pDSTS{}.
Instead of priorities, each vertex $v$ is assigned a weight $w(v)$ from some \emph{commutative semigroup} $(S,+,0)$.
Instead of \alg{tree-min} queries, we need to answer queries $\alg{tree-sum}(v)$, which return the sum of the weights of all vertices in the same component as $v$.
Since the set of priorities under taking minima forms
a commutative semigroup, this is a strict generalization of the DTM problem.
It turns out that our data structure can be generalized to solve the this problem as well.

\begin{restatable}{theorem}{restateDSTSMain}\label{p:dsts-main}
	There is a universally optimal data structure for the \pDSTS{} problem.
\end{restatable}

It should be noted that the semigroup itself is neither given as part of the input nor is it fixed; rather, data structures for \pDSTS{} are expected to work with arbitrary semigroups. That is, data structures may use the $+$ operation and the neutral element $0$, and may assume that associativity, etc., hold. This is sometimes called the \emph{semigroup model}~\cite{Yao1982a}.

In the \emph{group model}, where subtraction is possible, our data structure still works, but is no longer universally optimal. Indeed, computing group $\alg{tree-sum}$s on a star is possible in linear total time: Simply subtract $w(v)$ from the total weight whenever deleting an edge to a leaf $v$.

\paragraph{Incremental and offline variants.}
Finally, let us briefly consider the \emph{incremental} tree minima or tree semigroup sum problem. Adding an edge merges two trees into one, and the new minimum (or sum) can be computed in constant time. Thus, the problem boils down to finding the component containing a query vertex. This is the well-known \emph{union-find} problem, which can be solved in almost-linear time\footnote{In this paper, \emph{almost-constant} means one of the variants of the inverse Ackermann function, and \emph{almost-linear} means linear times almost-constant.}~\cite{Tarjan1975}. The offline variant of the incremental problem, where all operations are given in advance, can even be solved in linear time~\cite{GabowTarjan1983UnionFind}. Since the offline incremental and offline decremental variants are equivalent, offline \pDTM{} can also be solved in linear time.

\subsection{Techniques}\label{sec:tech}

In this section, we briefly present the main ingredients of our universally optimal data structure for \pDTM{} and the associated lower bound. For simplicity, we only assume the initial forest is a tree, the number of operations is linear in the number of vertices, and every edge is deleted eventually.

\begin{figure}
	\includegraphics[width=\textwidth]{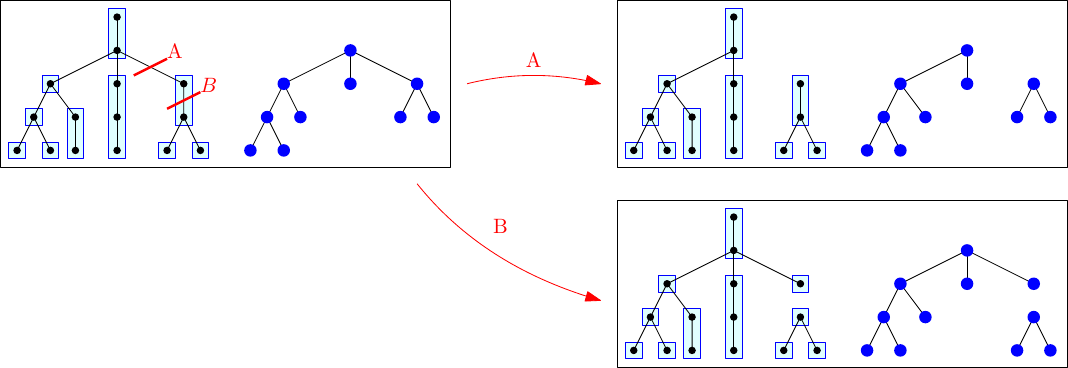}
	\caption{A tree partitioned into chains and the corresponding compressed tree (top left), and the effect an edge deletion between or within super-nodes (right). Chains are light blue rectangles, and super-nodes are blue.%
	}
	\label{fig:compr-tree-intro}
\end{figure}

\paragraph{Dynamic trees with chain compression.}
As mentioned above, classical dynamic tree data structures achieve a running time of $\fO( \log n )$ per operation, but for \emph{paths}, we can achieve $\fO(1)$ time via range minimum queries. The first idea is to combine these two data structures. To do this, we partition the initial tree into induced paths (called \emph{chains}), as shown in \cref{fig:compr-tree-intro}. For each of the chains, we maintain the constant-time DTM data structure. We also maintain the tree obtained by contracting each chain into a single node (called a \emph{super-node}), using a logarithmic-time DTM data structure $D$ based on dynamic forests.

The priority of each super-node is always the minimum of the associated chain; thus, a tree minimum can always be found with $D.\alg{tree-min}$. An edge deletion may affect up to one of the constant-time DTM data structures, and either deletes an edge in $D$ (see \cref{fig:compr-tree-intro}, top right) or ``splits'' a super-node in $D$ (see \cref{fig:compr-tree-intro}, bottom right). The latter operation is not standard in dynamic forests, but easy to implement.

Observe that when the initial tree has $\ell$ leaves, then it can be partitioned into $\fO(\ell)$ chains. Thus, the compressed tree stored in $D$ initially has $\fO(\ell)$ super-nodes. While ``splitting'' can create new super-nodes, it can be shown that the number of super-nodes in a single tree in $D$ cannot ever be more than $\fO(\ell)$. This implies that every dynamic tree operation takes $\fO(\log \ell)$ time. Since the operations on the chain data structures only take constant time, the overall running time for $\fO(n)$ operations is $\fO( n \log \ell )$. This is strictly better than $\Theta(n \log n)$ for some initial trees, but not universally optimal.

\paragraph{Lower bound.}
Before refining the data structure above, we describe our lower bound.
Recall that if the initial tree is a star, then we can sort the leaf priorities by repeatedly deleting the (edge to the) minimum-priority leaf. Actually, for any initial tree, we can sort its leaves this way, which implies a general $\Omega( \ell \log \ell )$ lower bound.
Thus, for trees with $\Theta(n)$ leaves, our previous algorithm is optimal. However, for some trees, neither the $\fO( n \log \ell )$ upper bound nor the $\Omega( \ell \log \ell )$ lower bound is tight.

A tight lower bound can be obtained as follows. Make the initial tree $T$ rooted by choosing an arbitrary root.
Suppose that for each node $v$ in $T$, each ancestor of $v$ has larger priority. This means that \alg{tree-min} will always return a leaf. Thus, we can again sort the nodes in $T$ by priority by repeatedly deleting the minimum-priority leaf. This is not \emph{general} sorting, since the set of possible priority assignments is restricted by the condition above. Instead, we get a special case of the aforementioned \emph{partial order sorting} problem. Taking $T$ as the Hasse diagram of a partial order $P$, then the possible orderings respecting the above condition is the set $\LEs(P)$ of linear extensions of $P$. By the information-theoretic lower bound, we need $\log |\LEs(P)|$ comparisons to sort, so we have an $\Omega( \log |\LEs(P)| )$ lower bound for $2n$ DTM operations on $T$.

\paragraph{Matching the lower bound with splay trees.}

Now, if we continue assuming that the minimum is always a leaf, then we can match the $\Omega(\log |\LEs(P)|)$ bound using known results on partial order sorting. In particular the optimal algorithm of Haeupler et al.~\cite{HaeuplerHladikEtAl2024c}, in our language, maintains the set of leaves in a special fast priority queue. It repeatedly extracts the minimum leaf from the priority queue, removes it from the tree, and then, if this makes another vertex a leaf, inserts that vertex into the priority queue.

For our data structure, we follow the same basic framework. However, we cannot use the same priority queue implementation: In the final data structure, general edge deletions will be allowed, which sometimes require \emph{splitting} the set of leaves in a certain way. To support this, we instead use a \emph{splay tree}~\cite{SleatorTarjan1985} that maintains the leaf set in left-to-right order, and supports finding the minimum as well as arbitrary splits.\footnote{A negative result for a related problem~\cite{HaeuplerHlad'ikEtAl2024} may seem to imply that splay trees are unsuitable for such a task; however, note that they only consider splay trees which keep nodes sorted by priority, which we avoid.} The analysis is very different.

\paragraph{Putting things together.}

We now go back to the original algorithm using dynamic forests and chain compressions. We modify it as follows. The dynamic tree $D$ is now no longer responsible for maintaining the priorities of the root and leaf super-nodes of the compressed forest; their priorities are permanently set to $\infty$ in $D$. The leaf super-nodes of each tree are stored in a splay tree, and the root super-node priority is maintained explicitly for each tree. %
We can now find the minimum of a tree by checking (1) the root super-node, (2) the corresponding tree in $D$, and (3) the leaf splay tree.

The running time analysis very roughly goes as follows. We always explicitly update tree minima when deleting edges, so \alg{tree-min} takes constant time. For edge deletions, we distinguish between three cases:
When we delete an edge within a root super-node, only that super-node needs to be updated, in constant time.
When we delete an edge within a leaf super-node, only the corresponding splay tree needs to be updated, in optimal amortized time by the above discussion.

The only remaining case are edge deletions within other super-nodes, or between super-nodes, which take (essentially) $\fO( \log \ell )$ amortized time to update $D$ and to split a splay tree. However, it can be shown that such edge deletions only happen $\fO( \ell )$ times overall. Thus, the total time for these edge deletions is $\fO( \ell \log \ell)$, which is accounted for by the basic $\Omega( \ell \log \ell )$ leaf sorting lower bound.

\paragraph{Semigroups.}
Most parts of our data structure generalize to the semigroup model without change. The only exception is the path-specific DTM implementation that uses range minimum queries. In fact, for the corresponding \emph{semigroup range query} problem, a super-linear lower bound is known~\cite{Yao1982a}.%

However, with a different data structure based on previous work by Gajewska and Tarjan~\cite{GajewskaTarjan1986} and Weiss~\cite{Weiss1994}, we can still solve the \pDSTS{} problem in linear time on paths. With this, we can generalize our main data structure to semigroup sums.

\paragraph{Cartesian trees and path-minimum queries.} We discussed above how to compute Cartesian trees using a DTM data structure. To prove this algorithm is universally optimal, we give an estimate of the information-theoretic lower bound, which again is based on partial order sorting.

For path-minimum queries, recall that we claimed that using a Cartesian tree together with an LCA data structure has universally optimal preprocessing time if no comparisons are allowed in \alg{path-min} queries. To prove this, we show that a Cartesian tree can be constructed using only \alg{path-min} queries. Although the number of queries needed is quite large, this shows that any \pPM{} data structure can be used to compute an elimination tree. If that data structure performs no comparisons in \alg{path-min} queries, all comparisons necessary to compute the elimination tree must occur in preprocessing, and thus our information-theoretic lower bound applies.

\subsection{Conclusions}

We defined universal optimality for data structures, and gave a universally optimal data structure for \pDTM{}.
Our definition could be applied to other data structure problems, not necessarily restricted to dynamic weighted graphs. In particular, any problem where both ``easy'' special cases and tight worst-case bounds are known could be a good candidate for universal optimality.
Our techniques for \pDTM{} are robust enough to allow replacing vertex weights with edge weights, and generalize tree minima queries to arbitrary semigroup tree sum queries.

We applied our data structures to obtain a universally optimal algorithm to compute Cartesian trees. An interesting generalization is \emph{weighted centroid trees}. Here, a vertex-weighted tree $(G,w)$ is given, and we need to construct the elimination tree on $G$ where the root of every subtree is the \emph{centroid} of the respective subgraph $H$ of $G$. A vertex $v$ is a centroid of $H$ if removing $v$ splits $H$ into components whose weight is at most half the total weight of $H$~\cite{KarivHakimi1979a,BerendsohnEtAl2023a}.

The special case where all weights are distinct powers of two is equivalent to Cartesian trees; thus, our lower bounds hold. However, our approach does not work for centroid trees; essentially, the problem of computing centroids does not decompose as nicely as the problem of computing minima. We leave open the question whether there exists a universally optimal algorithm for computing weighted centroid trees.

We used Cartesian trees to give a \emph{static} data structure for \pPM{}, with a restricted universal optimality guarantee, which holds only if all \alg{path-min} queries perform zero comparisons. If comparisons in \alg{path-min} queries are allowed, an inverse-Ackermann tradeoff between preprocessing and query time is possible and optimal~\cite{Chazelle1987,Pettie2006}. Since \pPM{} on \emph{paths} (i.e., the \emph{range minimum queries} problem) can be solved with linear preprocessing and constant query time, this is another candidate for universal optimality.

As a final note, observe this paper and several recent results~\cite{HaeuplerHlad'ikEtAl2024,HaeuplerHladikEtAl2025a,Berendsohn2025} use \emph{adaptive data structures} (splay trees here and working-set heaps in the other references). In this context, adaptivity means that these data structures are able to exploit certain \emph{sequential patterns} in access sequences. These results suggest that they are also useful to exploit the input graph structure, i.e., the \emph{topology} of an input.

\paragraph{Organization of this paper.}
In the remainder of this introduction, we discuss related work. In \cref{sec:prelims}, we fix some notation and define our problems and results more formally.
\Cref{sec:lb,sec:ds,sec:short-seqs} are devoted to the proof of our main result. We start with a simplified  lower bound for the special case of $\Theta(n)$ operations (\cref{sec:lb}) and then give a data structure that matches it (\cref{sec:ds}). Then, we refine the lower bound and running time analysis to show optimality for an arbitrary number of operations (\cref{sec:short-seqs}).
In \cref{sec:ct-pm}, we describe our two applications (the \pCT{} and \pPM{} problems). Finally, we discuss variants of our problems with semigroup weights (\cref{sec:semigroups}) and edge weights (\cref{sec:edge-weights}).
The appendix contains some omitted proofs.

\subsection{Related work}\label{sec:related}

In this section, we briefly review relevant literature that was not mentioned above.

\paragraph{Dynamic and semi-dynamic graph data structures.}

A \emph{dynamic graph data structure} maintains a changing graph, e.g., under edge insertions and deletions, while maintaining certain properties of the graph or allowing certain queries.

The special case of \emph{forests} is also well-studied, from the early \emph{link-cut trees} of Sleator and Tarjan~\cite{SleatorTarjan1983,SleatorTarjan1985}, to the more recent \emph{top trees} of Alstrup, Holm, De Lichtenberg, and Thorup~\cite{AlstrupEtAl2005,TarjanWerneck2005,HolmRotenbergEtAl2018}.
The latter provide a very flexible framework to support a variety of queries (including \alg{tree-min}) under edge insertions and deletions, in $\fO(\log n)$ time per operation. An earlier data structure for \alg{tree-min} queries in forests, based on link-cut trees, was given by Goldberg, Grigoriadis, and Tarjan~\cite{GoldbergGrigoriadisEtAl1991} for use in a maximum flow algorithm.

Semi-dynamic problems, either \emph{decremental} (only edge deletions) or \emph{incremental} (only edge insertions) often allow for better running times or simpler algorithms~\cite{Henzinger2018}. In the special case of forests, the running time per operation can drop from the usual $\Theta(\log n)$ to constant or almost-constant. For example, fully-dynamic connectivity (where we need to answer whether two vertices are connected or not) requires $\Omega(\log n)$ amortized time per operation on forests~\cite{PatrascuDemaine2006}. In contrast, the decremental variant requires only $\fO(1)$ time per operation~\cite{AlstrupSecherEtAl1997}, and the incremental variant can be solved with almost-constant (inverse-Ackermann) time per operation via the classical union-find data structure~\cite{Tarjan1975}.

An interesting problem that is closely related to \pDTM{} is the \emph{list splitting} problem~\cite{Gabow1985}. While not originally formulated in the language of graph theory, it is equivalent to DTM on \emph{collections of paths}, with an additional operation that allows decreasing the priority of a vertex. List splitting has several applications~\cite{Pettie2006}, and the best known data structure has almost-constant amortized running time per operation~\cite{PettieRamachandran2005}.
When we drop the priority-decreasing operation, we can solve the problem with linear preprocessing and constant amortized time per operation, using the aforementioned range minimum query data structure.

In another related paper, Brodal, Davoodi, and Rao~\cite{BrodalDavoodiEtAl2011} design data structures for \alg{path-min} queries (and their generalization to semigroups) in a dynamic tree, for a comparably limited set of graph-editing operations.

\paragraph{Cartesian trees and path minimum queries.}

Cartesian trees have been first studied in the special case where $G$ is a path~\cite{Vuillemin1980}. Recall that elimination trees on paths can be seen as binary trees on a sequence. The name \emph{Cartesian tree} comes from an interpretation of vertices as points in the plane, where index and priority are x- and y-coordinate, respectively.

Notable applications of Cartesian trees (on paths) are lowest common ancestor queries in rooted trees~\cite{BenderFarach-Colton2000}, as well as range minimum queries and their higher-dimensional equivalents~\cite{GabowBentleyEtAl1984,FischerHeun2006}.
Further, \emph{treaps}~\cite{AragonSeidel1989}, a well-known dynamic binary search tree data structure, are simply Cartesian trees with random priorities.

In the introduction, we defined generalizations of Cartesian trees to vertex-weighted and edge-weighted graphs. In particular the edge-weighted version (\emph{Cartesian EPTs}) restricted to \emph{trees} has been rediscovered several times, mainly for use in \alg{path-min} or \alg{bottleneck} query data structures~\cite{Chazelle1987,BoseMaheshwariEtAl2004,DemaineLandauEtAl2014}. There are further applications to clustering~\cite{Neto1999}, and the problem of computing Cartesian EPTs has been studied on its own~\cite{DeanMohan2013,DemaineLandauEtAl2014}. Specifically \alg{path-min} queries play an important role in the minimum spanning tree verification problem~\cite{Komlos1985,DixonRauchEtAl1992,King1995,Pettie2006}, and a generalization to semigroups has been studied~\cite{Tarjan1979,AlonSchieber1987}

The (vertex-weighted) \pCT{} problem on \emph{general graphs} is relevant to matrix factorization algorithms, and has been studied in the setting where the vertices are pre-sorted by priority~\cite{Liu1986,ZhuMutchler1994}.

\paragraph{Elimination trees and vertex search.}

Our vehicle to generalize Cartesian trees on paths to arbitrary graphs are \emph{elimination trees}, which have been intensely studied in their own right. They have been have been first defined in the study of matrix factorization~\cite{Liu1990}, and have strong connections to the concepts of \emph{tree-depth}~\cite{SanchezVillaamil2017} and \emph{tree-width}~\cite{BodlaenderEtAl1992}.\footnote{Note that there are multiple slightly different structures called \emph{elimination tree} in these contexts.}

Another, more recent application concerns a certain search model in trees. Given a tree, the task is to identify a vertex $v$ by \emph{querying} vertices. A query at $u$ reveals whether $u = v$, and, if $u \neq v$, reports the first edge on the path from $u$ to $v$. Here, an elimination tree serves as a data structure for tree search, similarly to a binary search tree for searching in sorted sequences.~\cite{DereniowskiNadolski2006,OnakParys2006,DereniowskiEtAl2017a,BoseCardinalEtAl2020,HoegemoBergougnouxEtAl2021,BerendsohnKozma2022a,BerendsohnEtAl2023a,Berendsohn2024b,DereniowskiWrosz2024,SadehKaplanEtAl2025}

In combinatorics, elimination trees define a class of polytopes called \emph{graph associahedra}~\cite{CarrDevadoss2004,Devadoss2009}, whose properties have been studied intensely in recent years~\cite{CarrDevadoss2004,Devadoss2009,CardinalEtAl2018b,CardinalEtAl2021a,CardinalPourninEtAl2021a,Berendsohn2022,ItoKakimuraEtAl2023,CardinalPourninEtAl2024,CunhaSauEtAl2025}.

\paragraph{Further results related to universal optimality.}
There are several more results that predate the term \emph{universal optimality}, but nonetheless fall under that definition in some sense. The \emph{partial order sorting} problem mentioned above had attracted some attention before, and several algorithms exist~\cite{Kislitsyn1968,Fredman1976,Linial1984,KahnKim1992,BrightwellFelsnerEtAl1995,CardinalFioriniEtAl2010,vanderHoogRutschmann2024a} that have worse running times, but nonetheless use an universally optimal \emph{number of comparisons}.
Also, this problem has recently been generalized from partial orders to so-called \emph{antimatroids}~\cite{Berendsohn2025}.

A somewhat related problem is \emph{local sorting}, where an undirected vertex-weighted graph is given, and each edge needs to be oriented towards the heavier endpoint. Goddard, King, and Schulman~\cite{GoddardKingEtAl1990} give a randomized algorithm that uses a universally optimal number of comparisons. Note that this problem can always be solved with one comparison per edge, and the interesting case is where fewer comparisons are needed. The \emph{running time} alone is not interesting in the context of universal optimality, since, trivially, it also is at least linear in the number of edges.

\section{Preliminaries}\label{sec:prelims}

In the following, we fix some notation and give basic graph-theoretic definitions. In \cref{sec:pre:dtm,sec:pre:comp-labeling,sec:pre:pm}, we define relevant data structure problems and make some observations.

All graphs in this paper are undirected. For a graph $G$, we let $V(G)$ denote the set of vertices and $E(G)$ the set of edges. If $e \in E(G)$, then $G \gminus e$ is the graph obtained by removing $e$. If $v \in V(G)$, then $G \gminus v$ is the graph obtained by removing $v$ and its incident edges.

\paragraph{Rooted forests.}
Internally, our algorithms and data structures usually operate on \emph{rooted} forests. We use the following notation for a rooted forest $F$. The sets of nodes is $V(F)$, the set of edges is $E(F)$, and we write $|F| = |V(F)|$ for short.
For each node $v \in V(F)$, we denote by $F_v$ the subtree rooted at $v$, i.e., the subtree induced by $v$ and all its descendants.

\begin{figure}
	\centering
	\includegraphics{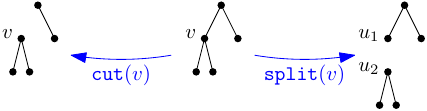}
	\caption{\alg{cut} and \alg{split} in a rooted tree.}\label{fig:cut-split}
\end{figure}

We define two operations that modify rooted forests; one for basic edge deletion, and one for internal use. \Cref{fig:cut-split} shows examples.

\begin{definition}\label{def:cut-split}
	Let $F$ be a rooted forest, and let $v \in V(F)$
	The operation $\alg{cut}(v)$ deletes the edge between $v$ and its parent, assuming that $v$ is not a root.
	The operation $\alg{split}(v)$ replaces $v$ by two new nodes $u_1, u_2$, such that $u_1$ inherits the parent of $v$, and $u_2$ inherits the children of $v$. The tuple $(u_1, u_2)$ is then returned by \alg{split}.
\end{definition}

Observe that \alg{split}, when applied to a leaf or root, does nothing but create
a new isolated node. This case will be easy to handle by our data structures, which makes the following lemma useful:

\begin{lemma}\label{p:num-cuts-splits}
	Let $F$ be a forest with $n$ nodes, and consider any valid sequence of \alg{cut}s and \alg{split}s applied to $F$. Then, the sequence contains no more than $n-1$ \alg{cut}s, and at most $n-2$ \alg{split}s are applied to non-root non-leaf nodes.
\end{lemma}
\begin{proof}
	Observe that no new edges are ever introduced.
	Since each \alg{cut} removes one of the up to $n-1$ edges, the first bound is immediate.
	For the second bound, observe that the number of non-root-non-leaf nodes is at most $n-2$ at the start. Each $\alg{split}(v)$ on a non-root non-leaf node $v$ replaces $v$ by a root and a leaf. Thus, there are at most $n-2$ such \alg{split}s.
\end{proof}

\paragraph{Priorities.}
In this paper, $\Prios$ denotes some totally ordered universe of \emph{priorities}. We assume priorities can be compared in constant time, but not accessed otherwise (see below). Our algorithms may internally use a special priority $\infty$ that is larger than all other priorities given in the input. A \emph{priority function} maps a set of elements (usually the vertices of a graph) to \emph{distinct} priorities.

\paragraph{Machine model.} We assume the standard Word RAM model with logarithmic word size and the usual modification for comparison-based algorithms, specified in the following. Priorities are not given explicitly. Instead, our algorithms have access to an oracle that answers queries of the form $p(v) < p(u)$, where $u$ and $v$ are vertices in the input graph (represented as integers), and $p$ denotes the input priority function.
Whenever an algorithm ``stores'' a priority somewhere, it really stores the \emph{node} with that priority. The special priority $\infty$ is stored as a special integer, and comparisons with it are resolved without calling the oracle.

\subsection{Decremental Tree Minima}\label{sec:pre:dtm}

For the rest of the paper, we define our central problem \problem{Decremental Tree Minima} as follows, using \emph{rooted forests} and the \alg{cut} operation defined above:

\begin{definition}\label{def:dtm}
	A data structure for \problem{Decremental Tree Minima} (DTM) maintains a rooted forest $F$ with a fixed priority function $p \colon V(F) \rightarrow \Prios$ under
	the following operations:
	\begin{itemize}
		\item $\alg{tree-min}(v) \rightarrow u$ returns the vertex $u \in V(G)$ with minimum priority $p(u)$ that is in the same tree as $v$.
		\item $\alg{cut}(v)$ removes the edge from $v$ to its parent (as in \cref{def:cut-split}).
	\end{itemize}
\end{definition}

It is easy to see that the unrooted variant is equivalent.
In the remainder of the paper, we will always assume that the initial forest $F^0$ is a \emph{tree}. If $F^0$ consists of several trees, we can build a separate data structure on each tree.\footnote{We can find the tree containing a query vertex with an auxiliary data structure described in \cref{sec:comp-labeling}.}
Our lower bound, due to its information-theoretic nature, is easily generalized to arbitrary initial forests. Thus, all of our optimality results carry over to the general forest case.

\subsection{Decremental Tree Roots}\label{sec:comp-labeling}\label{sec:pre:comp-labeling}\label{sec:pre:roots}

In this section, we describe a data structure for rooted forests that maintains a mapping from each node $v$ to the root of the tree containing $v$. This will be useful if we want to store some information about a \emph{tree} (like the minimum-priority vertex), since we can store that information in the root.

\begin{definition}
	A data structure for the \pDTR{} problem maintains a rooted forest $F$ under
	the following operations:
	\begin{itemize}
		\item $\alg{root}(v) \rightarrow r$ returns the root $r$ of the tree containing $v$:
		\item $\alg{cut}(v)$ removes the edge from $v$ to its parent (as in \cref{def:cut-split})
		\item $\alg{split}(v) \rightarrow (u_1, u_2)$ splits $v$ into two new nodes $u_1, u_2$ (see \cref{def:cut-split}).
	\end{itemize}
\end{definition}

It was observed before~\cite{AlstrupHusfeldtEtAl1998a,DemaineLandauEtAl2014} that a data structure by Alstrup, Secher, and Spork~\cite{AlstrupSecherEtAl1997} can be adapted to partially solve this problem with linear preprocessing time and constant time per operation. This data structure supports the \alg{root} and \alg{cut} operations.
Below, we show how to also support the \alg{split} operation.

\begin{lemma}
	There is a data structure for \pDTR{} that performs $m$ operations with an initial forest on $n$ nodes in $\fO(m+n)$ time.
\end{lemma}
\begin{proof}
	To support \alg{split}s, we first preprocess the initial forest $F^0$ as follows. For each node $v$, we apply a \alg{split} operation to $v$, obtaining a node $p_v$ that inherits the parent of $v$, and a node $c_v$ that inherits the children. Call $V(F^0)$ the \emph{original} nodes, and let $F^1$ be the modified forest.
	
	Now initialize the data structure mentioned above (without support for \alg{split}) with the modified forest $F^1$. Call this data structure $X$.
	At any given point, let $F$ denote the forest we are maintaining (initially $F = F^0$), and let $F'$ denote the modified forest maintained by $X$ (initially $F' = F^1$). Besides $X$, we also explicitly maintain $F$, and keep track of which nodes in $F$ are original. (New non-original nodes will be introduced every time we perform a \alg{split} in $F$.)
	
	We maintain following two invariants. First, at any point, $F'$ is precisely the forest obtained from $F$ by applying the above \alg{split} operation for each original node still contained in $F$. Second, every non-original node in $F$ is a root or leaf.
	
	We now show how to implement the three operations. Consider $\alg{cut}(v)$. If $v$ is original, we call $X.\alg{cut}(p_v)$. Otherwise, we call $X.\alg{cut}(v)$.
	
	Consider now $\alg{split}(v)$. If $v$ is original, we call $X.\alg{cut}(c_v)$, and return $(p_v, c_v)$. Otherwise, $v$ must be a root or leaf. We can thus simply create a new isolated node $u$ and return either $(u,v)$ (if $v$ is a root) or $(v,u)$ (if $v$ is a leaf).
	
	Finally, consider $\alg{root}(v)$. if $v$ is original, we compute $r \gets X.\alg{root}(p_v)$; otherwise, we compute $r \gets X.\alg{root}(v)$.
	Then, if $r = p_u$ for some original node $u$ that is contained in~$F$, we return $u$. Otherwise, we return $r$.
	
	Each operation clearly takes $\fO(1)$ time. Further, it is easy to see that \alg{cut} and \alg{split} maintain the invariants, and thus \alg{root} is correct.
\end{proof}

\subsection{Path Minima and Decremental Tree Minima on paths}\label{sec:pre:pm}

We now define the \emph{static} data structure for path minimum queries; here, it is more convenient to use unrooted trees.

\begin{definition}
	A data structure for \pPM{} (PM) represents a (static) tree with a priority function $p$ and supports the following query:
	\begin{itemize}
		\item $\alg{path-min}(v_1,v_2) \rightarrow u$ returns the vertex $u$ on the path between $v_1$ and $v_2$ with minimum priority $p(u)$.
	\end{itemize}
\end{definition}

If the input tree $G$ is a path, then the \problem{Path Minima} problem is equivalent to the well-known \problem{Range Minimum Query} (RMQ) problem (RMQ is usually defined on a \emph{sequence} of comparable elements instead of a path with comparable vertex priorities). RMQ is known to be solvable with $\fO(n)$ preprocessing time and $\fO(1)$ query time~\cite{BenderFarach-Colton2000}.

\begin{lemma}\label{p:pm-paths}
	Let $G$ be a path on $n$ vertices with a priority function $p$. Then the \pPM{} problem on $G$ can be solved with $\fO(n)$ preprocessing time and $\fO(1)$ worst-case time per operation.
\end{lemma}

This extends to the DTM problem in the same special case of the initial tree being a path.

\begin{corollary}\label{p:dtm-single-path}
	Let $T$ be a path with $n$ nodes rooted at one of its endpoints. Then the \pDTM{} problem with initial tree $T$ can be solved with $\fO(n)$ preprocessing time and $\fO(1)$ worst-case time per operation.
\end{corollary}
\begin{proof}
	We maintain the following structures. First, a \pDTR{} data structure $R$ (\cref{sec:comp-labeling}). Second, a \pPM{} data structure $M$ on the initial path~$T$. Note that both structures can be build with $\fO(n)$ preprocessing time and can execute every operation in $\fO(1)$ amortized time.
	
	At any point during execution, for each root $r$ of a component $C$ of the current forest, we maintain a pointer to the unique leaf $\ell_r$ of $C$.
	
	Whenever $\alg{tree-min}(v)$ is called, we first compute $r \gets R.\alg{root}(v)$. We then answer the query with $M.\alg{path-min}(r,\ell_r)$.
	
	Now consider a call $\alg{cut}(v)$. First, determine the parent $u$ of $v$, and compute $r \gets R.\alg{root}(v)$. The \alg{cut} splits the path from $\ell_r$ to $r$ into two paths with endpoints $\ell_r, v$ and $u, r$, respectively. Thus, we set $\ell_v \gets \ell_r$ and then $\ell_r \gets u$. Then, we can call $R.\alg{cut}(v)$ and $M.\alg{cut}(v)$.
\end{proof}

Finally, \cref{p:dtm-single-path} can be extended to a \emph{collection} of paths:

\begin{corollary}\label{p:dtm-paths}
	Let $F$ be a rooted forest on $n$ nodes where each component is a path rooted at one of its endpoints. Then the \pDTM{} problem with initial forest $F$ can be solved with $\fO(n)$ preprocessing time and $\fO(1)$ worst-case time per operation.
\end{corollary}
\begin{proof}
	Start by stringing together the components of $F$ into a single long path $P$. Then, initialize the data structure of \cref{p:dtm-single-path}, and cut $P$ into the original paths. Now proceed as in \cref{p:dtm-single-path}. The additional initialization time is clearly $\fO(n)$.
\end{proof}

\section{A basic lower bound via tree sorting}\label{sec:lb}

In this section, we show a lower bound for the \pDTM{} problem for operation sequences of sufficient length. Note that this is not sufficient to prove universal optimality, since we ignore shorter sequences; these are treated later in \cref{sec:short-seqs}. %

The main idea, reducing from partial order sorting, was already presented in the introduction. We start here by defining our special case of the partial order sorting problem.
Given a rooted tree $T$, we say a priority function $p$ on $V(T)$ is \emph{$T$-monotone} if for each node $v$ and each ancestor $u$ of~$v$ in $T$, we have $p(v) < p(u)$.

\begin{definition}\label{def:tree-sorting}
	In the \pTS{} problem, we are given a rooted tree $T$ and a $T$-monotone priority function $p$, and need to sort $V(T)$ by priority.
\end{definition}

Let $\fP_T$ be the set of $T$-monotone priority functions. For each priority function $p$ on $T$, let $\pi_p$ denote the ascending ordering of $V(T)$ induced by $p$. Let $\LEs(T) = \{ \pi_p \mid p \in \fP_T \}$.

Observe that $\LEs(T)$ is the set of linear extensions of the partial order with Hasse diagram $T$.
Also, $\LEs(T)$ is precisely the set of possible outputs of the \pTS{} problem on the tree $T$. Since we can only distinguish between different outputs with comparisons, the information-theoretic lower bound yields:

\begin{observation}
	Fix a rooted tree $T$. Then, for every \pTS{} algorithm~$A$, there exists a priority function $p$ such that $A$ requires at least $\log |\LEs(T)|$ comparisons on input $(T,p)$.
\end{observation}

It turns out that $\LEs(T)$ is characterized a simple formula. For a rooted tree $T$ on $n$ nodes, define the \emph{tree entropy} of $T$ as 
\[ H(T) = \sum_{v \in V(T)} \log \tfrac{n}{|T_v|}. \]

\begin{lemma}[{Knuth~\cite[Exercise 5.4.1.20]{Knuth1998}}]\label{p:tree-les}
	Let $T$ be a rooted tree on $n$ nodes. Then
	\[ |\LEs(T)| = |T|! \cdot \prod_{v \in V(T)} \frac{1}{|T_v|}, \]
	implying that $\log |\LEs(T)| \ge H(T) - n \log e$.
\end{lemma}

The formula in \cref{p:tree-les} is known in the literature as \emph{Knuth's hook-length formula}.
We prove a generalization of \cref{p:tree-les} in \cref{sec:short-seqs} (\cref{p:num-partial-tree-orders}).

We now turn back to the DTM problem.
Suppose we have a DTM data structure $X$ and are given an input $(T,p)$ for \pTS{}. Find the root $r$ of $T$, initialize $X$ with $(T,p)$, then repeatedly call $v \gets X.\alg{tree-min}(r)$ and $X.\alg{cut}(v)$, until all edges are deleted. Since $p$ is $T$-monotone, each node returned by \alg{tree-min} is a leaf, and the sequence of outputs of \alg{tree-min} is precisely $V(T)$ in sorted order.

Thus, a DTM data structure can solve \pTS{} with a sequence of $2n-2$ operations, so the $\log |\LEs(T)| \ge H(T) - n \log e$ lower bound applies. We also have a trivial $\Omega(n)$ bound, if the first operation is \alg{tree-min}. Therefore, we obtain:

\begin{corollary}\label{p:lb-dtm-longseq}
	For each DTM data structure $X$ and each rooted tree $T$ with $n \ge 2$ nodes, there exists a priority function $p$ on $T$ and a sequence $\sigma$ of $2n-2$ operations such that $\Time(X,T,p,\sigma) \ge \Omega( H(T) )$.
\end{corollary}

In the next section, our goal will be to give a data structure that matches \cref{p:lb-dtm-longseq} (which, as mentioned, is not quite enough for our definition of universal optimality). The following simple lower bounds will be useful:

\begin{lemma}\label{p:lb-leaves}\label{p:entropy-lbs}
	Let $T$ be a rooted tree with $n$ nodes and $\ell$ leaves. Then $H(T) \ge \ell \log n$ and $H(T) \ge n-1$.
\end{lemma}
\begin{proof}
	Observe that for each leaf $v$, we have $|S_v| = 1$, so $\log \tfrac{n}{|S_v|} = \log n$. Summing up only the contribution of the leaves to $H(T)$, we get $\ell \log n$, as required.
	
	We prove the second inequality by induction. if $n = 1$, then $H(T) = 0 = n-1$. Now suppose $n \ge 2$, and let $T^1, T^2, \dots, T^k$ be the trees obtained by removing the root from~$T$. Observe that $T^i_v = T_v$ for each $v \in V(T_i)$, and let $n_i = |T^i|$.
	With this, we have
	\begin{align*}
		H(T) = \sum_{v \in V(T)} \log \tfrac{n}{|T_v|} = \sum_{i=1}^k \sum_{v \in V(T^i)} \log \left(\tfrac{n}{n_i} \cdot \tfrac{n_i}{|T^i_v|}\right) = \sum_{i=1}^k n_i \log \tfrac{n}{n_i} + H(T^i).
	\end{align*}
	Since $n_i \le n-1$, we have $n^{n_i} \ge (n_i+1)^{n_i}$. It is easy to check that $(n_i+1)^{n_i} \ge 2 n_i^{n_i}$ by the binomial theorem, so $n_i \log \tfrac{n}{n_i} \ge 1$. The statement follows by induction.
\end{proof}

\section{Step-by-step towards universal optimality}\label{sec:ds}

In this section, we present our data structure for \pDTM{}.
For now, we will prove optimal running time for every initial tree $T$, for a worst-case sequence of operations of length at least $2|V(T)| - 2$ (matching \cref{p:lb-dtm-longseq}). 
We start with a relatively naive attempt, and then gradually improve our algorithm.
The full proof of universal optimality (for an arbitrary number of operations) is given in \cref{sec:short-seqs}.

\subsection{An \texorpdfstring{$\fO(n \log n)$}{O(n log n)} algorithm with dynamic forests}\label{sec:trees:dt-trivial}\label{sec:ds:dt-trivial}
A \emph{dynamic forest} data structure~\cite{SleatorTarjan1983,SleatorTarjan1985,AlstrupEtAl2005} maintains a rooted forest and some associated data along with edge insertion and deletion.
We use the following specification.
Note that for our $\fO(n \log n)$ \pDTM{} data structure, we only need the operations $\alg{cut}$ and \alg{find-min}, and only need $\fO(\log n)$ amortized upper bounds for their running time. The operations \alg{set-priority} and \alg{split} will be useful later.

\begin{definition}\label{def:dt}
	A \emph{dynamic forest data structure} represents a rooted forest $F$ with vertex priorities $p \colon V(F) \rightarrow \Prios$. It supports the operations listed below. %
	\begin{itemize}
		\item $D.\alg{find-min}() \rightarrow v$ returns the node $v$ with minimum priority. %
		
		\item $D.\alg{cut}(v)$ removes the edge between $v$ and its parent (see \cref{def:cut-split}).

		\item $D.\alg{set-priority}(v,q)$ sets $p(v) \gets q$. %
		
		\item $D.\alg{split}(v) \rightarrow (u_1, u_2)$ splits $v$ into two new nodes (see \cref{def:cut-split}). The priorities of the two new nodes $u_1$ and $u_2$ are initialized as $p(v)$.
	\end{itemize}
\end{definition}

\begin{lemma}\label{p:dt-running-times}
	There is a dynamic forest data structure such that:
	\begin{itemize}
		\item Initialization takes $\fO(n)$ time, where $n$ is the number of nodes in the given initial forest;
		\item \alg{find-min} takes $\fO(1)$ time;
		\item $\alg{split}(v)$ takes $\fO(1)$ time when $v$ is a root or a leaf;
		\item \alg{split} in every other case, as well as \alg{cut} and \alg{set-priority}, all take $\fO(1 + \log n)$ time, where $n$ is the number of nodes in the tree containing the affected nodes.
	\end{itemize}
	All running times are amortized.
\end{lemma}
\begin{proof}
	All operations besides \alg{split} can be implemented in $\fO(\log n)$ amortized time by standard solutions~\cite{SleatorTarjan1983,SleatorTarjan1985,GoldbergGrigoriadisEtAl1991}. Suppose we have such a data structure $X$, which also supports the operation $X.\alg{link}(u,v)$ that makes $u$ a child of $v$.
	
	We can implement $\alg{split}(v)$ by first calling $X.\alg{cut}(v)$, then adding a new child $u_1$ to the former parent of $v$ using $X.\alg{link}(u_1, v)$. Then, we set the priority of $u_1$ to $p(v)$ and finally return $(u_1, v)$. This yields the required running time of $\fO(\log n)$ in the general case.
	
	To implement $\alg{find-min}$ in $\fO(1)$ time, we store the minimum value of each component $S$ of $F$ explicitly in the root $r$ of $S$, as a value $q_r$. We maintain a \pDTR{} data structure on $F$ to access $r$. Whenever we need to update $q_r$, we compute it with the standard $\fO(\log n)$-time \alg{find-min} operation of $X$. This happens at most twice after each \alg{cut} or \alg{split}, so we stay within the $\fO(\log n)$ budget.
	
	Finally, we describe how to perform \alg{split} in $\fO(1)$ time for roots and leaves.	
	Suppose $v$ is a leaf. For $\alg{split}(v)$ we create a new, disconnected vertex $u_2$ with priority~$p(v)$, and return $(v, u_2)$. This requires only calling $\alg{set-priority}(u_2, p(v))$, which takes $\fO(1)$ time since $u_2$ is an isolated vertex.
	Moreover, the only newly created root is $u_2$, we clearly have $q_{u_2} = p(v)$, and $q_r$ does not need updating for any other root $r$.
	
	If $v$ is a root, we proceed similarly.
\end{proof}

\Cref{p:dt-running-times} implies that, starting with a tree with $n$ nodes, we can perform $m$ DTM operations in $\fO(m \log n + n)$ time.

\subsection{Chain compression}\label{sec:ds:chains}

In this section, we give an improved data structure with $\fO( \log \ell)$ running time per operation, where $\ell$ is the number of leaves in the input forest $F$. The main idea is that since we can solve the DTM problem efficiently on \emph{paths}, we identify long induced paths in $F$, and solve the problem separately on them.

Let $F$ be a rooted forest. A \emph{chain} in $F$ is a sequence $C = (v_1, v_2, \dots, v_k)$ such that for each $i \in [k-1]$, the unique child of $v_i$ is $v_{i+1}$. Note that $v_k$ may have an arbitrary number of children. The \emph{top node} of the chain is $v_1$, and the \emph{bottom node} is $v_k$. \emph{Contracting} $C$ yields a new rooted forest $F'$, where $C$ is replaced by a single node $x$, such that the parent of $x$ in $F'$ is the parent of $v_1$ in $F$ (if it exists), and the children of $x$ in $F'$ are the children of $v_k$ in $F$.

\begin{figure}
	\centering
	\begin{minipage}{.3\textwidth}
		\centering\includegraphics{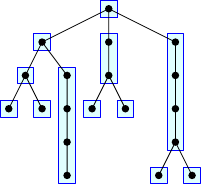}
	\end{minipage}
	\begin{minipage}{.3\textwidth}
		\centering\includegraphics{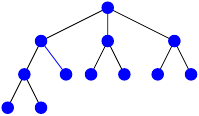}
	\end{minipage}
	\begin{minipage}{.3\textwidth}
		\centering\includegraphics{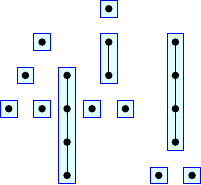}
	\end{minipage}\\
	\vspace{3.5mm}
	\begin{minipage}{.3\textwidth}
		\centering
		Original forest $F$
	\end{minipage}
	\begin{minipage}{.3\textwidth}
		\centering
		Compressed forest $F'$
	\end{minipage}
	\begin{minipage}{.3\textwidth}
		\centering
		Chain forest $F_C$
	\end{minipage}
	\caption{A maximal-chain compression of a forest, with priorities omitted. Chains are overlaid in light blue.%
	}\label{fig:compression}
\end{figure}

Let $F'$ be a rooted forest obtained from $F$ by contracting a pairwise disjoint set of chains that covers the whole forest (see \cref{fig:compression}). For each $x \in V(F')$, let $C(x)$ be the chain in $F$ that produced $x$ (observe that $C(x)$ may consist of only one node). If $p \colon V(F) \rightarrow \Prios$ is a vertex priority function for $F$, then we define $p' \colon V(F') \rightarrow \Prios$ as $p'(x) = \min_{v \in C(x)} p(v)$. Together, the tuple $(F',p',C)$ is called a \emph{compression} of $(F,p)$. In the following, we call the nodes of a compression \emph{super-nodes}. We will sometimes use chains and super-nodes interchangably when there is no danger of confusion, e.g., we may say for $v \in V(F)$ and $x \in V(F')$ that ``$v$ is contained in $x$'' or ``$v$ is the top node of $x$''. We call leaves in $F'$ \emph{super-leaves} and roots in $F'$ \emph{super-roots}.

It will also be useful to consider the collection of chains as a rooted forest, where each tree is a chain $C(x)$. We denote this \emph{chain forest} by $F_C$ (see \cref{fig:compression}, right).

The key idea for our improved data structure is that if we contract all \emph{maximal} chains, all non-leaf super-nodes in the resulting compression have at least two children, so the number of super-nodes is linear in the number of leaves. If we build a dynamic tree on the compression, each operation has running time $\fO(\log \ell)$, as desired.

\paragraph{Maintaining a compression.}
We´ now show how to maintain the structure of a compression under edge deletions, disregarding priorities. Define an \emph{unweighted compression} $(F',C)$ of a forest $F$ in the obvious way. Suppose we perform a $\alg{cut}(v)$. Then $(F',C)$ can be updated as follows (see \cref{fig:compr-tree-intro}): If $v$ and and its parent $u$ are in distinct chains $C(x)$ and $C(y)$, we also perform a $\alg{cut}$ in~$F'$ (see \cref{fig:canonical-cut-split}, left). If $u$ and $v$ are in a common chain $C(x)$, we need to perform a \alg{split} operation on $x$, obtaining two new nodes $y_1, y_2$ (see \cref{fig:canonical-cut-split}, right). Also, we need to split the chain $C(x)$ into two new chains $C(y_1), C(y_2)$. The result is clearly a valid compression of the new forest.

We refer to this operation on $F'$ as the \emph{canonical} \alg{cut} or \alg{split} associated to the operation $\alg{cut}(u)$ in $F$.
We now give a data structure maintaining an unweighted compression and some useful additional data.

\begin{figure}
	\centering
	\includegraphics{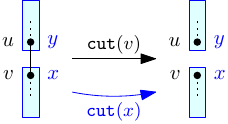}
	\hspace{15mm}
	\includegraphics{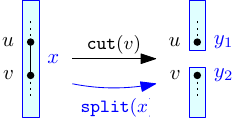}
	\caption{A \alg{cut} operation in $F$ causing a canonical \alg{cut} (left) or \alg{split}(right) in the compression $F'$.}\label{fig:canonical-cut-split}
\end{figure}

\begin{lemma}\label{p:compresssion-ds}
	There is a data structure that maintains an unweighted compression $(F', C)$ of a forest $F$ with $n$ vertices under $\alg{cut}$ operations in $F$, with $\fO(n)$ preprocessing time, and constant time per \alg{cut}. Each \alg{cut} in $F$ performs the canonical \alg{cut} or \alg{split} in $(F',C)$. Moreover, the data structure maintains the following auxiliary structures:
	\begin{itemize}
		\item An explicit representation of $F$ and $F'$, with parent pointers and child lists.
		\item A constant-time mapping from each node $v \in V(F)$ to the super-node $x \in V(F)$ that contains~$v$.
		\item A constant-time mapping from each node $v \in V(F)$ to the root $r \in V(F)$ of the component in $F$ containing $v$.
		\item A constant-time mapping from each super-node $x$ to the top node of its chain $C(x)$.
	\end{itemize}
\end{lemma}
\begin{proof}
	In addition to the explicit representations of $F$ and $F'$, we maintain two \pDTR{} data structures:
	One, denoted $L_F$, on $F$, and one, denoted $L_C$, on $F_C$. Recall that $F_C$ is the forest of all chains $C(x)$, so $L_C.\alg{root}(v)$ returns the top node in the chain containing $v$.
	Moreover, we explicitly maintain a pointer $t(x)$ from each super-node $x$ to the top node $v$ of its chain $C(x)$, and a reverse pointer $t^{-1}(v) = x$.
	
	Now we can find the super-node containing a given node $v$ in constant time by computing $t^{-1}(L_C.\alg{root}(v))$. The other two mappings are provided by $L_F$ and $t$.
	
	Suppose we call $\alg{cut}(v)$, and say $u$ is the parent of $v$. We can easily update the explicit representation of $F$ directly, and call $L_F.\alg{cut}(v)$.
	
	To update $F'$, $L_C$, $t$, and $t'$, we first determine whether we need to perform a \alg{cut} or a \alg{split}. For this, we check whether $u$ and $v$ are in the same chain, by checking equality of $L_C.\alg{root}(u)$ and $L_C.\alg{root}(v)$.
	
	If both are in the the same chain $C(x)$, we need to split this chain.
	We first call $L_C.\alg{cut}(v)$. Then, we perform a $\alg{split}(x)$ in $F'$, which creates nodes $(y_1, y_2)$. We set $t(y_1) = t(x)$ and $t(y_2) = v$, and update $t^{-1}$ accordingly.
	
	Suppose now $u$ is in a chain $C(y_1)$ and $v$ is in a different chain $C(y_2)$. Then $y_1$ is the parent of $y_2$ in $F'$. We simply perform a $\alg{cut}(y_2)$ in $F'$. In this case $F_C$ and $t$ do not need to be modified, since no chains change.
\end{proof}

\paragraph{Tree minima with chain compression.}

We now describe our improved data structure for the DTM problem. An instance maintains a compression $(F', p', C)$ on $(F,p)$. At the start, the compression is obtained by contracting all maximal chains in the initial tree~$T^0$.

We maintain $F$ and $(F',C)$ in an instance $A$ of the auxiliary structure of \cref{p:compresssion-ds}, and additionally:
\begin{itemize}
	\item A dynamic forest $D$ for $(F',p')$.
	\item A Path-DTM data structure $R$ for $F_C$, using \cref{p:dtm-paths}.
\end{itemize}

Observe that, given an input $(F,p)$, we can initialize the data structure in linear time.
We now show how to implement the two DTM operations. Let $X$ denote an instance of our data structure.
For $X.\alg{tree-min}(v)$, we simply defer to $D.\alg{tree-min}(x)$, where $x$ is the super-node that contains $v$. This returns the correct answer because, for every component of $F$ and its associated component of $F'$, we have $\min_{x \in V(F')} p'(x) = \min_{v \in V(F)} p(v)$.

Consider now a call $X.\alg{cut}(v)$, when $u$ is the parent of $v$ in $F$.
Perform $\alg{cut}(v)$ in~$A$, thereby identifying the corresponding canonical \alg{cut} or \alg{split} in~$F'$. We perform the respective operation in $D$. Also, if $v$ and its parent $u$ are in the same chain (i.e., if the canonical operation is a \alg{split}), we call $R.\alg{cut}(v)$. Otherwise, $F_C$ does not change, so $R$ does not need to be updated.

It remains to update the priorities of super-nodes in $D$. This is only necessary if the canonical operation is a \alg{split}, since otherwise no super-nodes are changed.
Let $y_1$ be the super-node containing $u$, and let $y_2$ be the super-node containing $v$, after the operation (see \cref{fig:canonical-cut-split}, right). We update their priorities with $D.\alg{set-priority}(y_1, R.\alg{tree-min}(u))$ and $D.\alg{set-priority}(y_2, R.\alg{tree-min}(v))$. Since now $C(y_1)$ and $C(y_2)$ are the two parts of the previous chain $C(x)$ that contain $u$ and~$v$, respectively, this correctly assigns the new priorities $p'(y_1)$ and $p'(y_2)$.
Observe that every other super-node $x \in V(F') \setminus \{y_1, y_2\}$ already existed before the operation, and its priority $p'(x)$ does not change. Thus, no further updates are necessary.

\paragraph{Running time analysis.}

Initializing each data structure needs $\fO(n)$ time, and the running time of $\texttt{tree-min}$ is clearly constant.
In $\alg{cut}(v)$, the only non-constant operations are the dynamic tree operations, so the running time is $\fO(\log n')$, where $n'$ is the number of super-nodes in the component $S'$ of $F$ containing $x$. We now show that $n' \le 2\ell-1$, where $\ell$ is the number of leaves in the \emph{initial} forest.

As we have argued before, since we contract all maximal chains, every inner node in the initial super-forest $F'$ has two children, implying that $|F'| \le 2 \ell - 1$. This implies in particular that $|E(F')| \le 2 \ell - 2$. Observe that \alg{cut}s and \alg{split}s cannot \emph{increase} the number of edges. Thus, the compression never has more than $2 \ell - 2$ edges, implying that each of its components has no more than $2 \ell - 1$ vertices.
With this, we have:

\begin{proposition}\label{p:ub-logl}
	The \pDTM{} problem can be solved with linear-time preprocessing, constant worst-case time per $\alg{tree-min}$, and $\fO(\log \ell)$ amortized time per $\alg{cut}$, where $\ell$ is the number of leaves of the initial forest $F$.
\end{proposition}

Unfortunately, \cref{p:ub-logl} does not imply universal optimality, even in the restricted case of $\Theta(n)$ operations. Indeed, consider the input tree $T$ where the root $r$ has $k+1$ child subtrees, one being a chain of $n-k-1$ nodes, and each other consisting of a single node (see \cref{fig:bad-leaf-tree}). It is easy to check that the lower bound given by \cref{p:lb-dtm-longseq} is $\Theta(n + k \log k)$. Thus, if $k \approx n / \log n$, the lower bound is linear, but the upper bound from \cref{p:ub-logl} is $\Theta(n \log \ell) = \Theta(n \log k) = \Theta(n \log n)$.

\begin{figure}
	\centering
	\includegraphics{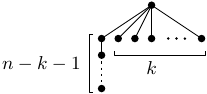}
	\caption{An initial tree where \cref{p:ub-logl} is not optimal.}\label{fig:bad-leaf-tree}
\end{figure}

A more thorough analysis reveals that the problem lies with modifications of the \emph{super-roots} and \emph{super-leaves}. Consider an operation $\alg{cut}(v)$, and call it \emph{type 1} if both $v$ and its parent $u$ are both contained in a single super-leaf or super-root. Otherwise, call it \emph{type 2}. For each type-2 operation, the corresponding canonical super-forest operation is either a \alg{cut}, or a \alg{split} of a non-root-non-leaf super-node. By \cref{p:num-cuts-splits}, there can be at most linearly many such operations. Since the super-forest has $\fO(\ell)$ nodes, and each operation takes at most $\fO(\log \ell)$ time by the analysis above, the total time of type-2 operations is $\fO( \ell \log \ell)$, which is within our budget (see \cref{p:lb-leaves}).

For type-1 operations, we need a different approach.
It turns out that we can handle \alg{split}s of \emph{super-roots} without much effort, with constant time per operation. We now sketch the solution as a warm-up. More details and the solution for super-leaves are presented in \cref{sec:trees:splay}.

The idea is to simply \emph{not} store the priority for super-roots in $D$. This means setting the priority of every super-root in $D$ to $\infty$.
Observe that we can still compute $X.\alg{tree-min}(v)$: First, if $x$ is the super-node containing $v$, then $D.\alg{tree-min}(x)$ will return the minimum node in the component $S$ containing $x$, among all nodes that are \emph{not} in the ``root chain''.
Using the auxiliary data structure, we can compute the root node~$r$ of the component containing $v$. Then $R.\alg{tree-min}(r)$ is the minimum node within the root chain. Thus, $X.\alg{tree-min}(v) = \min( D.\alg{tree-min}(x), R.\alg{tree-min}(r))$.

We can ensure that each root super-node has priority $\infty$ in $D$ by simply calling $D.\alg{set-priority}$ whenever a \alg{cut} or a type-2 \alg{split} occurs. This requires $\fO(\log \ell)$ time, which is subsumed by the other dynamic tree operations.

Finally, we argue that $X.\alg{cut}(u)$ only needs constant time if $u$ is in a root chain $C(x)$. First, observe that $D.\alg{split}(x)$ only takes constant time by \cref{p:dt-running-times}. Second, the two returned super-nodes $y_1$ and $y_2$ are both super-roots, and $\alg{split}$ initializes them with priority $p(x) = \infty$. Thus, no priority updates are necessary.

\subsection{Leaf handling with splay trees}\label{sec:trees:splay}

To deal with leaves, we need to finally consider the details of our lower bound in \cref{p:lb-dtm-longseq}. We first describe the solution informally. As we did with root nodes, we now relieve the dynamic tree $D$ of the duty of maintaining super-leaf priorities and set all of them to $\infty$. Like we argued for super-roots, this means that, for $D$, \alg{split}ting super-leafs costs only constant time, while \alg{cut}s and \alg{split}s of non-leaf non-root nodes still costs $\fO( \log \ell )$ time. By \cref{p:num-cuts-splits}, the latter can happen at most $\fO(\ell)$ times. For our $\Theta(n)$ operations, the total running time for \alg{cut}s and \alg{split}s is thus $\fO( n + \ell \log \ell)$, which is within the required bound (\cref{p:lb-leaves}).

The problem are now the \alg{tree-min} operations. The reason root nodes are easy to handle is that there is only one root per component in $F$; obviously, this is not the case for leaves. To handle super-leaves, we use a data structure based on \emph{splay trees}~\cite{SleatorTarjan1985} that maintains an ordered sequence of super-leaves for each component in $F$. The running time analysis uses a potential argument based on the \emph{access lemma}~\cite{SleatorTarjan1985} to match \cref{p:lb-dtm-longseq}.

We now proceed with describing the data structure. We begin with some technical details on how to maintain an \emph{ordered} forest. Then, we describe two new sub-structures, and finally put everything together and analyze the running time.

\paragraph{Ordered forests and compressions.}
From now on, we assume that for each node of $F$, its children are ordered in a fixed, but arbitrary way. We call $F$ an \emph{ordered} forest. This induces an ordering on the leaves of each component $T$ of $F$, which will be useful later. A sub-forest $S$ inherits the ordering in the obvious way, and so does a compression of $F$ or $S$.

When maintaining an explicit representation of an ordered forest $F$, each node maintains pointers to its left and right sibling, if either exists (i.e., each child list is represented as a \emph{linked list}). This allows $\alg{cut}$ and $\alg{split}$ operations to maintain the ordering in constant time. With this, it is easy to adapt the auxiliary structure for compressions (\cref{p:compresssion-ds}) to ordered forests.

\subsubsection{Maintaining leaf sequences}

We now describe our data structure responsible for maintaining leaf priorities. Essentially, the sequence of super-leaves of $S'$ is maintained in a \emph{splay tree}~\cite{SleatorTarjan1985}, which allows fast modification of super-leaves.
Note that $\alg{split}$ and $\alg{cut}$ may split the super-leaf sequence in a certain way. This can be done with standard algorithms, though we need an additional data structure that determines exactly \emph{where} we have to split the leaf sequence (discussed in \cref{sec:leftmost-leaves}).

We proceed with the specification of our data structure. For technical reasons, we describe it differently than the other data structures in this paper. Instead of maintaining each leaf sequence of each component of $F$ in a single data structure, we have a data structure for each individual tree. Applying a split operation produces two new data structures for the two new sequences, and destroys the old one.

\newcommand{\pSSM}{\problem{Splittable Sequence Minima}}

\begin{definition}\label{def:ssm}
	A data structure for the \pSSM{} (SSM) problem maintains a sequence of elements $x_1, x_2, \dots, x_n$ from a universe $U$, which is equipped with a priority function $p \colon U \rightarrow \Prios$.
	It supports the following operations:
	\begin{itemize}
		\item $\texttt{min}() \rightarrow x$ returns the minimum-priority element $x$ that is \texttt{visible}.
		
		\item $\texttt{replace}( x, y )$ replaces the element $x = x_i$ with some $y \in U$.
		
		\item $\texttt{split-interval}( x, y; z ) \rightarrow (Y_1,Y_2)$: Let $z \in U \cup \{\bot\}$, and assume $x = x_i$ and $y = x_j$ with $i < j$.
		Destroys $X$ and returns two new instances: $Y_1$ on $x_1, \dots, x_{i-1}, z, x_{j+1}, \dots x_n$, where $z$ is omitted if $z = \bot$, and $Y_2$ on $x_i, \dots, x_j$. (See \cref{fig:split-interval-seq}.)
	\end{itemize}
\end{definition}

\begin{figure}
	\centering
	\includegraphics{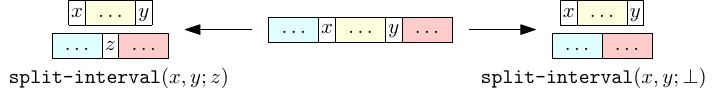}
	\caption{A \alg{split-interval} operation.}\label{fig:split-interval-seq}
\end{figure}

Observe that a single element $x$ can be removed with $X.\alg{split-interval}(x,x,\bot)$, although this operation will be quite costly. The following lemma gives the running times achieved by our splay-based implementation. For the analysis, we will define a weight function $w$ for the universe of possible values in the data structure. The running time bounds below work for an arbitrary weight function $w$, even though the algorithm does not know $w$.\footnote{It is important here that all running times are amortized; $w$ determines how the total cost is ``distributed'' among the different operations for the analysis.}

\begin{lemma}\label{p:ssm-ds}
	Let $U$ be a universe equipped with a priority function $p \colon U \rightarrow \Prios$ and a weight function $w \colon U \rightarrow \N_+$. Assume the $p$ can be computed in constant time ($w$ may be unknown and is only used in the analysis). Then, there exists a data structure for \pSSM{} over $U$, with the following amortized running times:
	\begin{itemize}
		\item $\fO(n + W)$ for preprocessing, when initialized with $n$ elements of total weight $W$.
		\item $\fO(1)$ time for $\alg{min}$.
		\item $\fO( \log \tfrac{W}{w(x)} )$ for $\alg{replace}( x, y )$, where $W$ denotes the total weight \emph{after} the operation.
		\item $\fO( \log W )$ for $\alg{split-interval}$, where $W$ denotes the total weight of all elements in both instances produced by the operation.
	\end{itemize}
\end{lemma}

We now review splay trees, which are the base for our implementation of SSM.
Originally~\cite{SleatorTarjan1985}, splay trees were presented as an implementation of \emph{ordered dictionaries}. They maintain an ordered list of $n$ keys, or key-value pairs, in a binary search tree (BST).
As usual, the keys are stored in the tree in \emph{symmetric order}, that is, an in-order traversal of the BST visits the keys in the same order as they appear in the maintained list.

Other standard BST data structures, such as AVL trees~\cite{Adelson-VelskyLandis1962},
keep the height of the BST at $\fO(\log n)$, to ensure $\fO(\log n)$ worst-case search time.
In contrast, splay trees do not provide such a guarantee. Instead, after every search, they \emph{rebalance} themselves in order to speed up future searches. While only an \emph{amortized} $\fO(\log n)$ running time per search is ensured by splay trees, they have the upside of a simple implementation without additional bookkeeping, and also can be shown to perform \emph{better} then $\fO(\log n)$ per operation when the search sequence has a certain structure.

The central operation for splay trees is the eponymous $\texttt{splay}(x)$, which moves a node $x$ to the root of the splay tree through a sequence of \emph{rotations}. It is performed after every search, where $x$ is the found node. \texttt{splay} is also useful to implement insertion and deletion of keys. For definitions of rotations and the \alg{splay} operation, we refer to Sleator and Tarjan~\cite{SleatorTarjan1985}. For us, the only important detail is the rotations (and thus \alg{splay}) do not change the symmetric order of the nodes.

Here, we use splay trees in a different way; we never need to search for keys, but instead use it as a dynamic decomposition of a sequence as in \cref{def:ssm}. The nodes will maintain minima of contiguous subsequences, and \alg{splay} will be useful to implement the operations \alg{replace} and \alg{split-interval}. There are other examples of a similar ``structural'' use of splay trees, most prominently the dynamic forest data structure of Sleator and Tarjan~\cite{SleatorTarjan1983,SleatorTarjan1985}.

Here, we consider a splay tree to be a binary tree $T$ with three allowed operations:
\begin{itemize}
	\item $\alg{splay}(x)$, which brings $x$ to the root of $T$ with a series of rotations. The exact way this is done is not relevant here, see \cite{SleatorTarjan1985} for details.
	\item Attach the root $r'$ of another splay tree $T'$ as a left or right child to the root $r$ of~$T$.
	\item Detach a child from the root $r$ of $T$, producing a new splay tree.
\end{itemize}

\paragraph{Analysis of splay trees.}
For the analysis of our data structure, we will consider a \emph{set} of splay trees, called a \emph{splay forest}. Each node $x$ comes from a universe $U$ as in \cref{p:ssm-ds}, and thus has an associated weight $w(x) \in \N_+$. The weight of a subtree $T_x$, denoted $w(T_x)$, is the sum of weights of its nodes, i.e., $w(T_x) = \sum_{y \in V(T_x)} w(y)$. The \emph{rank} of a node is $r(x) = \log w(T_x)$, and the \emph{potential} of a splay tree $T$ is $\Phi(T) = \sum_{x \in V(T)} r(x)$. The potential of a splay forest is the sum of the splay tree potentials.

Consider a call to $\alg{splay}(x)$ that changes the splay tree containing $x$ from $T$ to $T'$. The \emph{amortized running time} of $\alg{splay}(x)$ is the actual running time plus the potential change $\Phi(T') - \Phi(T)$.
Sleator and Tarjan's well-known \emph{access lemma}~\cite[lemma 1]{SleatorTarjan1985} states that the amortized running time is $\fO( 1 + \log \frac{W}{w(x)} )$, where $W$ is the sum of weights of all tree nodes.\footnote{The potential needs to be scaled by a constant factor depending on the actual running time of a single \emph{rotation}; we omit this detail for simplicity.}
The overall running time of a sequence of $\texttt{splay}$ operations is then at most the total amortized running time plus the potential $\Phi(F_0)$ of the initial splay forest $F_0$.

As mentioned above, we can also attach a root node as a child to another node, and detach a child from its parent. Both these operations change the potential. We will take this into account when analyzing the amortized running time of the SSM operations.

\paragraph{Implementing SSM.}
An instance of our SSM implementation consists of a single splay tree $T$, containing the elements (keys) $x_1, x_2, \dots, x_n \in U$ in symmetric order. We assume that we can find the predecessor $x_{i-1}$ and the successor $x_{i+1}$ of an element $x_i$ in constant time (should either exist). This can be accomplished by explicitly maintaining pointers within~$T$.

For each node $x \in V(T)$, we additionally maintain the minimum priority $q_x$ of all nodes in $T_x$, i.e.:
\[ q_x = \min_{y \in V(T_x)} p(y). \]
Whenever we perform a rotation, attach a child to the root, or detach a child from the root, we only need to recompute $q_x$ for at most two nodes. Clearly, $q_x$ can be recomputed from the values $q_y$ of its children in constant time. Thus, the time to update $q_x$ is dominated by other operations, and will be ignored from now on.

We now proceed with the implementation and running time analysis of each SSM operation. In the following, $W$ denotes the sum of all node weights before the operation, and $W'$ denotes the sum of all node weights after the operation, in all produced trees.
\begin{itemize}
	\item Initialization is a little subtle, since we need to be careful about the initial potential of the splay tree, which will be charged to the amortized running time. In the literature, usually no assumption is made on the initial tree. Since the rank of each node is at most $\log W'$, the total potential of any tree can be upper bounded by $n \log W'$.
	However, when starting with a \emph{complete} binary tree, the initial potential is $\fO(W')$.\footnote{Recall that we assume that all weights are at least one. Sleator and Tarjan~\cite{SleatorTarjan1985} allow arbitrary positive weights; in that case, the $\fO(W')$ bound does not necessarily hold.} We defer the proof to \cref{app:tree-pot}.
	
	Computing each $q_x$ can be done in $\fO(n)$ time with a bottom-up traversal. The total amortized running time is thus $\fO( n + W') = \fO(W')$, since weights are positive integers.
	
	\item \texttt{min} simply returns the value $q_x$ of the root $x$ of the splay tree, without changing anything. This is the minimum of all node priorities by definition of $q_x$.
	
	\item $\texttt{replace}(x, y)$ first \alg{splay}s $x$ to the root, and then simply replaces $x$ with $y$. Again, we can compute $q_y$ in constant time. The replacement changes total weight in the tree from $W$ to $W' = W + w(y) - w(x)$, and thus increases the potential by $\log \tfrac{W'}{W}$. The overall amortized running time is thus $\fO( 1 + \log \frac{W}{w(x)} + \log \frac{W'}{W}) = \fO( 1 + \log \frac{W'}{w(x)})$.

	\begin{figure}
		\centering
		\includegraphics{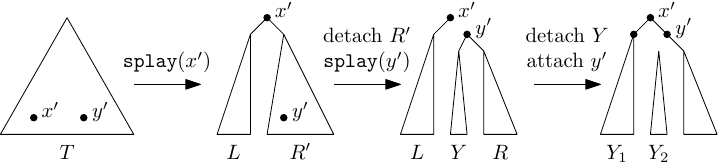}
		\caption{A call $\alg{split-interval}(x, y; \bot)$.}\label{fig:split-interval}
	\end{figure}
	
	\item $\texttt{split-interval}( x, y; z )$ first computes the predecessor $x'$ of $x$ and the successor $y'$ of $y$. We assume both exist in the following; the procedure is easily adapted otherwise.
	
	Start by \texttt{splay}ing $x'$ (see \cref{fig:split-interval}) and detach the right subtree $R'$ of $x'$. Then \alg{splay} $y'	$ (in $R'$) and detach its left subtree $Y$. We have now split our data structure into three trees $L, Y, R$, where $Y$ contains the nodes strictly between $x'$ and $y'$, i.e., the interval from $x$ to $y$. Simply return $Y$ as $Y_2$. If $z = \bot$, attach $R$ (with root $y'$) as the right subtree to the root $x$ of $L$, and return the result as $Y_1$. If $z \neq \bot$, create a new node $z$, attach $L$ and $R$ as left and right subtrees, and return the result as~$Y_1$.
	
	Each \texttt{splay} runs in $\fO( 1 + \log W )$ amortized time, since all weights are positive integers. The subtree attachment and potential creation of $z$ increases the potential by at most $\log W$ if $z = \bot$, or by at most $\log W + \log (W + w(z)) \le 2 \log W'$ otherwise. Also, note that predecessor and successor pointers only change for $x$, $x'$, $y$, $y'$, and (potentially) $z$, and can be updated in constant time.
\end{itemize}

This proves \cref{p:ssm-ds}.

\subsubsection{Maintaining leftmost and rightmost leaves}\label{sec:leftmost-leaves}

When performing a $\alg{cut}$ or $\alg{split}$ in the super-forest, we may need to split a leaf sequence with a $\alg{split-interval}(x, y; z)$ operation on the \pSSM{} data structure. To do this, we need to identify the appropriate leaves $x$ and $y$. We use the following data structure for this task.

\newcommand{\pEL}{\problem{Extremal Leaves}}

\begin{definition}\label{def:el-ds}
	An \pEL{} data structure maintains an ordered rooted forest $F$ under the operations \alg{cut} and \alg{split}, and supports the following query:
	\begin{itemize}
		\item $\alg{extremal-leaves}(v) \rightarrow (u_1,u_2)$ returns the leftmost leaf $u_1$ and the rightmost leaf $u_2$ in the subtree $T_v$.
	\end{itemize}
\end{definition}

\begin{figure}
	\centering
	\includegraphics{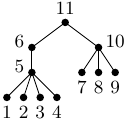}
	\caption{A left-to-right post-order numbering of an ordered rooted tree.}\label{fig:post-order}
\end{figure}

\begin{lemma}\label{p:ds-ext-leaves}
	An \pEL{} data structure on a rooted forest with initially $n$ nodes can be implemented with $\fO(n)$ preprocessing time,
	$\fO(1)$ amortized running time for $\alg{split}(x)$ if $x$ is a root or leaf, and $\fO(\log n)$ amortized running time for each other operation.
\end{lemma}
\begin{proof}
	We show how to maintain \emph{leftmost} leaves; rightmost leaves can be handled in a symmetric way. At the start, assign each node $v$ a unique label $\ell_v \in [n]$ via a left-to-right post-order traversal, in linear time (see \cref{fig:post-order}).
	More precisely, for two nodes $u, v$ we have $\ell_u < \ell_v$ if and only if either $u$ is a descendant of $v$, or, if $a$ is the lowest common ancestor of $u$ and $v$, the child subtree $T_c$ of $a$ containing $u$ is to the left of the child subtree $T_{c'}$ of $a$ containing $v$.
	
	We maintain nodes and their labels in a dynamic forest data structure $D$ (see \cref{def:dt}), with labels as priorities. Recall that this data structure supports $\alg{cut}$ and $\alg{split}$ in the required time (\cref{p:dt-running-times}). It additionally supports an operation $D.\alg{link}(u,v)$ that makes $u$ a child of $v$ in $\fO(\log n)$ time (see the proof of \cref{p:dt-running-times}).
	
	Consider a forest $F$ obtained after a series of $\alg{cut}$ and $\alg{split}$ operations, and take a component $T$ of $F$. Observe that the labeling of $T$ still satisfies the condition above. This implies that the leftmost leaf in a rooted subtree $T_v$ is precisely the minimum-label node in $T_v$. Thus, we can find the leftmost leaf in $T_v$ by calling first $D.\alg{cut}(v)$, then $D.\alg{tree-min}(v)$, and then $D.\alg{link}(v,u)$, where $u$ is the parent of $v$. If $v$ is a root, we only need to call $D.\alg{tree-min}(v)$. In any case, the running time is $\fO(\log n)$.
\end{proof}

\subsubsection{Finishing the main data structure}\label{sec:ds:main-finish}

We now give the final, universally optimal data structure. It maintains an ordered compresssion $(F',C,p')$ of the forest $(F,p)$, initialized by contracting all maximal chains,
and consists of:
\begin{itemize}
	\item The auxiliary structure $A$ of \cref{p:compresssion-ds} for $F$ and $(F',C)$.
	
	\item A Path-DTM data structure $R$ for $F_C$, using \cref{p:dtm-paths}.
	
	\item A dynamic forest $D$ that maintains $(F',p')$, except that the priority of every super-root and every super-leaf is $\infty$ in $D$. Note that $D$ does \emph{not} maintain the child order of nodes $F'$.
	
	\item An \pEL{} data structure $E$ for $F'$.
	
	\item For each component $S'$ of $F'$ with $|V(S')| \ge 2$, a \pSSM{} data structure containing all super-leaves of $S'$, in order.
	The universe $U$ is the set of super-leaves that are created at any point in the algorithm. The priority of each super-leaf $x$ is defined as $p'(x)$, and its weight (only for the analysis) is defined as $w(x) = \max_{v \in C(x)} |T^0_v|$, where $T^0$ is the initial tree.
	
	We denote this data structure by $L_x$, where $x$ is the super-root of $S'$. Each super-root $x$ in $F'$ maintains a pointer to $L_x$.
\end{itemize}

Note that the priority of every element in any $L_x$ can be computed in constant time (as required by \cref{p:ssm-ds}). Indeed, we have $p'(x) = R.\alg{min}(v)$, where $v$ is the top node in $x$, which can be found in constant time using $A$.

Let $X$ be an instance of the main data structure.
We now show how to implement the two operations $X.\alg{tree-min}$ and $X.\alg{cut}$.

For $X.\alg{tree-min}(v)$, we first determine the root $r$ of the component $S$ of $F$ containing~$v$, using $A$. Then, we compute $R.\alg{tree-min}(r)$, which reports the minimum node in the root chain of $S$.
Second, we determine the super-node $x$ such that $r$ is contained in $x$, again using $A$. Note that $x$ is the super-root of the super-tree $S'$ corresponding to~$S$. Thus, it holds a pointer to the \pSSM{} data structure $L_x$ containing the leaves of $S'$, and we can determine their minimum with $L_x.\alg{min}()$.
Finally, we call $D.\alg{tree-min}(x)$.

We have now covered the super-root, super-leaves, and other super-nodes of $S'$. Thus, taking the minimum of these three will return the minimum node in $S$.
If $|V(S')| = 1$, then $L_x$ does not exist, but computing $R.\alg{tree-min}(r)$ will already give us the minimum node in~$S$.

We proceed with $X.\alg{cut}(v)$. Let $y$ be the super-nodes containing $v$, let $S$ be the component of $F$ containing $v$, let $S'$ be the corresponding component of $F'$, and let $x$ be the super-root of $S'$ (all before the operation). We can determine $x$ and $y$ using $A$, as above.

We update $A$ and $R$ by directly calling $A.\alg{cut}(v)$ and $R.\alg{cut}(v)$, the latter only if $v$ and its parent are in the same chain. %
The call $A.\alg{cut}(v)$ will result in a canonical \alg{cut} or \alg{split} in $F'$.

Now apply the canonical operation to $E$ and $D$. This may turn up to two super-nodes into roots or leaves; we set the priority of these nodes to $\infty$ in $D$ using $D.\alg{set-priority}$ (if the priority was not already $\infty$).

We now show how to update the SSM data structures. Note that the only affected such structure is $L_x$. We consider several different cases:
\begin{enumerate}[(a)]
	\item Suppose $|S'| = 1$: Then $L_x$ does not exist (since we only maintain $L_x$ for super-components of size at least two). The canonical operation must be a $\alg{split}(x)$, and it creates two super-components consisting of single nodes. Thus, no SSM structures need to be created.
	
	\begin{figure}
		\centering
		\includegraphics{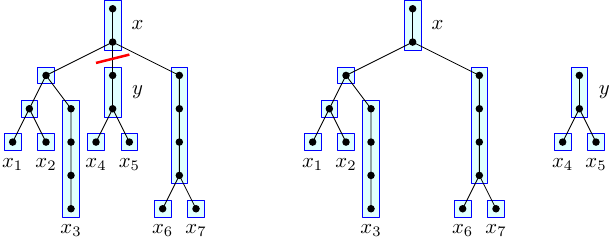}
		\caption{Updating $L_x$ in case \ref{item:update-ssm-cut}, a canonical \alg{cut}, with $\alg{split-interval}(x_4, x_5; \bot)$.}\label{fig:update-ssm-cut}
	\end{figure}
	
	\item\label{item:update-ssm-cut} Suppose the canonical operation is $\alg{cut}(y)$ (see \cref{fig:update-ssm-cut}). In this case, we need to split $L_x$ into two parts, one containing all leaves that are descendants of $y$, and one part containing all remaining leaves. To do this, we first compute $(z_1,z_2) \gets E.\alg{extremal-leaves}(y)$, then call $(L_1, L_2) \gets L_x.\alg{split-interval}(z_1, z_2; \bot)$, and finally set $L_x \gets L_1$ and $L_y \gets L_2$. If $y$ is a super-leaf (and thus now forms component of size 1), we instead discard $L_2$.
	
	\item Suppose the canonical operation is $\alg{split}(y)$, where $y$ is a super-root. Let $(y_1, y_2)$ be the two super-nodes resulting from the split. Note that $y_1$ is the only super-node in its super-component, so we do not compute a SSM data structure for it. For $y_2$, we can simply set $L_{y_2} \gets L_y$; the leaf sequence does not change.
	
	\item Suppose the canonical operation is $\alg{split}(y)$, where $y$ is a super-leaf. Let $(y_1, y_2)$ be the two super-nodes resulting from the split. Now $y_2$ is the only super-node in its super-component, so we do not compute a SSM data structure for it. On the other hand, we call $L_x.\alg{replace}(y, y_1)$.
	
	\begin{figure}
		\centering
		\includegraphics{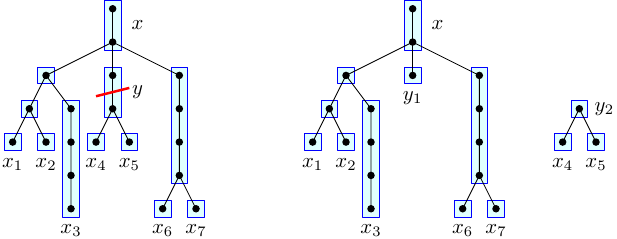}
		\caption{Updating $L_x$ in case \ref{item:update-ssm-split}, a canonical \alg{split} on a non-leaf-non-root super-node, with $\alg{split-interval}(x_3, x_4; y_1)$.}\label{fig:update-ssm-split}
	\end{figure}
	
	\item\label{item:update-ssm-split} Finally, suppose the canonical operation is $\alg{split}(y)$, where $y$ is neither a super-leaf nor a super-root (see \cref{fig:update-ssm-split}). Let $(y_1, y_2)$ two super-nodes resulting from the split. We proceed similarly to the $\alg{cut}$ case: Compute $(z_1,z_2) \gets E.\alg{extremal-leaves}(y)$, call $(L_1, L_2) \gets L_x.\alg{split-interval}(z_1, z_2; y_1)$, and then set $L_x \gets L_1$ and $L_{y_2} \gets L_2$.
\end{enumerate}

\paragraph{Running time.}
In the following, let $T^0$ denote the initial tree, let $n = |T^0|$.
Recall that we let the universe $U$ for our SSM structures be the set of all super-nodes ever created during the operation of the algorithm, and for each $x \in U$, we let $w(x) = \max_{v \in C(x)} |T^0_v|$. (Note that we consider super-nodes to be immutable for this purpose; every $\alg{split}$ destroys one super-node and creates two new ones. In particular, $C(x)$ is the unique chain associated to $x$ during its lifetime.)

Let $W^*$ be the maximum total weight of all super-nodes contained in a single SSM structure $L_x$, at any point.  The following lemma will be useful:

\begin{lemma}\label{p:rt-total-weights}
	$W^* \le |T^0|$.
\end{lemma}
\begin{proof}
	Consider some component $T'$ of the super-forest at some point during the algorithm. For each super-leaf $x$ of $T'$, let $v_x$ be the maximum-weight node in $C(x)$, so that $w(x) = w(v_x) = |T^0_{v_x}|$. Let $Q$ be the set of all nodes $v_x$.
	It is easy to see that all nodes in $Q$ are \emph{independent} in the initial tree $T^0$; that is, no $v \in Q$ is an ancestor of another $v' \in Q$. Thus, taking all descendants of all $v \in Q$ covers each node in $T^0$ at most once, so $\sum_{v \in Q} |T^0_v| \le |T^0|$. Since this argument applies to every component $T'$ at any point during the algorithm, we have $W^* \le |T^0|$.
\end{proof}

We are now ready to analyze the running time of each operation.

\begin{lemma}\label{p:ds-init-time}
	Initialization takes $\fO(n)$ amortized time.
\end{lemma}
\begin{proof}
	Initialization is $\fO(n)$ for each structure, except for the SSM structure $L_x$, which requires $\fO(n + W^*)$ time. By \cref{p:rt-total-weights}, this is $\fO(n)$.
\end{proof}

\begin{lemma}\label{p:ds-tree-min-time}
	$X.\alg{tree-min}(v)$ takes constant amortized time.
\end{lemma}
\begin{proof}
	We have three calls $R.\alg{tree-min}$, $D.\alg{tree-min}$, and $L_x.\alg{tree-min}$, each of which take constant time, and then simply compute the minimum of these three.
\end{proof}

\begin{lemma}\label{p:ds-cut-time}
	Let $v$ be a node, and let $y$ be the super-node containing $v$ in the current super-forest. $X.\alg{cut}(v)$ takes
	\begin{itemize}
		\item $\fO( \log \tfrac{n}{|T^0_v|} )$ time if $y$ is a super-leaf and the canonical super-forest operation for $\alg{cut}(v)$ is $\alg{split}(y)$;
		\item $\fO(1)$ time if $y$ is a super-root; and
		\item $\fO( \log n)$ time otherwise.
	\end{itemize}
\end{lemma}
\begin{proof}
	Updating $A$ and $R$ is done in constant time in all cases. Updating $D$ is always at most $\fO(\log \ell) \le \fO( \log n )$.
	
	Suppose now that the canonical super-forest operation corresponding to $\alg{cut}(v)$ is $\alg{split}(y)$, and $y$ is a super-leaf. Then $D.\alg{split}(y)$ and $E.\alg{split}(y)$ take $\fO(1)$ time by \cref{p:dt-running-times,p:ds-ext-leaves}. Also, since both created nodes are leaves, we do not update their priorities in $D$ (which are initialized to $\infty$). Thus, updating $D$ and $E$ takes $\fO(1)$ overall time. The same is true if $y$ is a super-root.
	
	Now let $x$ be the super-root of the super-tree containing $y$, and consider updates to $L_x$.
	At worst, we call $E.\alg{extremal-leaves}$, $L_x.\alg{split-interval}$, and $L_x.\alg{replace}$, with running time $\fO( \log n + \log W^* ) = \fO( \log n)$ by \cref{p:rt-total-weights}.
	Now first suppose that $y$ is a super-root. No changes to any structure $L_x$ are made in this case, with constant additional work.
	
	Finally, suppose that $y$ is a super-leaf. Then the running time is dominated by $L_x.\alg{replace}(y, y_1)$, with running time bounded by $\fO(\log \frac{W^*}{w(y)})$. Recall that $w(y) = |T^0_u|$, where $u$ is the heaviest node in $C(y)$.
	Since $v$ is also contained in $y$, we have $w(y) = |T^0_u| \ge |T^0_v|$. Thus, using $W^* \le \fO(n)$ by \cref{p:rt-total-weights}, we get the desired $\fO( \log \tfrac{n}{|T^0_v|} )$.
\end{proof}

Combining \cref{p:ds-init-time,p:ds-tree-min-time,p:ds-cut-time} yields:

\begin{theorem}\label{p:uopt-partial}\label{p:ub-partial}
	Initializing our DTM data structure with a rooted tree $T^0$, any priority function, and then performing $m$ operations on it takes $\fO(m + H(T^0))$ total time.
\end{theorem}
\begin{proof}
	Each node of $F$ is \alg{cut} at most once.
	Consider first the $\alg{cut}(v)$ operations with a canonical \alg{cut} or a canonical \alg{split} of a non-root-non-leaf super-node. By \cref{p:num-cuts-splits}, there are $\fO(\ell)$ operations, for a total running time of $\fO(\ell \log n)$ by \cref{p:ds-cut-time}. By \cref{p:lb-leaves}, this is bounded by $\fO(H(T^0))$.
	
	For all other \alg{cut}s, the total amortized time is \[ \sum_{v \in V(F)} \fO\left(\log \tfrac{|T^0|}{|T^0_v|}\right) = \fO(H(T^0)). \]
	Each of the up to $m$ \alg{tree-min} operations takes $\fO(1)$ time. Finally, initialization takes $\fO(n)$ time, and $H(T^0) \ge n-1$ by \cref{p:entropy-lbs}.
\end{proof}

This yields partial universal optimality, for operation sequences of length $m \ge 2n-2	$. In the next section, we show that our data structure is actually universally optimal even on shorter operation sequences.

\section{Universal optimality for short sequences}\label{sec:short-seqs}

In this section, we extend \cref{p:uopt-partial} and \cref{p:lb-dtm-longseq} to prove full universal optimality of our data structure, including for operation sequences of length $m < 2n-2$. Note that no changes to the data structure are necessary; only the analysis needs to be refined.

We first quantify the optimal running time, then prove our data structure matches it (\cref{sec:short-ub}), and finally prove that it is a lower bound for all data structures (\cref{sec:short-lb}). We need the following generalizations of the \emph{tree entropy} defined in \cref{sec:lb}.

Let $T$ be a rooted tree and $S \subseteq V(T)$. Then we define
\[ H_S(T) = \sum_{v \in S} \log \frac{|S|}{|V(T_v) \cap S|}. \]
Further, for $k \in \N_+$, we define
\[ H_k(T) = \max_{\substack{S \subseteq V(T)\\|S| \le k}} H_S(T). \]

Observe that for $k \ge |T|$, we have $H_k(T) = H_{V(T)}(T) = H(T)$.
We now prove a few useful properties of $H_k(T)$.

\begin{lemma}\label{p:lb-leaves-general}
	Let $T$ be a rooted tree with $n \ge 2$ nodes and $\ell$ leaves, and let $k \in \N_+$.
	Then $H_k(T) + n \ge \tfrac13 \min\{k, \ell\} \log n$.
\end{lemma}
\begin{proof}
	Let $m = \min\{k,\ell\}$, let $S'$ be a set of $m$ leaves of $T$, and let $S$ be an arbitrary vertex set of size $k$ with $S' \subseteq S \subseteq V(T)$. Observe that
	\[ H_k(T) \ge H_S(T) \ge H_{S'}(T) = m \log m. \]
	Our claim holds trivially when $m \le n / \log n$. If $m \ge n / \log n$, then $H_k(T) \ge m \log m \ge m (\log n - \log\log n) \ge \tfrac13 m \log n$, as desired.
\end{proof}

\begin{lemma}\label{p:Hm-desc-formula}
	Let $T$ be a rooted tree, and let $S \subseteq V(T)$ be closed under taking descendants. Then $H_S(T) = \sum_{v \in S} \log \tfrac{|S|}{|T_v|}$.
\end{lemma}
\begin{proof}
	For each $v \in S$, we have $V(T_v) \subseteq S$ by assumption; the statement then follows by definition of $H_S(T)$.
\end{proof}

\begin{lemma}\label{p:Hm-ass-closed-le}
	Let $T$ be a rooted tree and $k \in \N_+$. Then there exists a set $S \subseteq V(T)$ of size at most $k$ that is closed under taking descendants and satisfies $H_k(T) = H_S(T)$.
\end{lemma}
\begin{proof}
	Let $S \subseteq V(T)$ with $|S| \le k$ and $H_k(T) = H_S(T)$. We claim that, if some $u \in S$ has a child $v \notin S$, we have $H_S(T) \le H_{S'}(T)$ for $S' = S \setminus \{u\} \cup \{v\}$. Thus, repeatedly replacing $u$ with $v$ will eventually yield a set that satisfies the lemma.
	
	Towards our claim, observe that we have $|V(T_w) \cap S| = |V(T_w) \cap S'|$ for all $w \in S \setminus \{u\}$. Thus,
	\[ H_{S'}(T) - H_S(T) = \log \frac{k}{|V(T_v) \cap S'|} - \log \frac{k}{|V(T_u) \cap S|}. \]

	It is easy to see that $|V(T_v) \cap S'| \le |V(T_u) \cap S|$, which implies $H_{S'}(T) - H_S(T) \ge 0$, proving our claim.
\end{proof}

\begin{corollary}\label{p:Hm-ass-closed}
	For $k \le |T|$, \cref{p:Hm-ass-closed-le} holds even when requiring that $S$ has size exactly~$k$.
\end{corollary}
\begin{proof}
	Suppose that \cref{p:Hm-ass-closed-le} yields a set $S$ of size $\ell < k$ with $H_S(T) = H_k(T)$. Since $S$ is closed under taking descendants and $|S| < k \le n$, the root $r$ of $T$ is not contained in $S$. Consider now the set $S' = S \cup \{r\}$. We have
	\begin{align*}
		H_{S'}(T) = \sum_{v \in S'} \log \frac{|S'|}{V(T_v) \cap S'}
			&= \log \frac{|S'|}{V(T_r) \cap S'} + \sum_{v \in S} \log \frac{|S'|}{V(T_v) \cap S'}\\
			&= \log \frac{\ell+1}{\ell+1} + \sum_{v \in S} \log \frac{\ell+1}{V(T_v) \cap S} > H_S(T).
	\end{align*}
	This implies $H_S(T) < H_k(T)$, a contradiction.
\end{proof}

\begin{lemma}\label{p:Hm-mon}
	Let $T$ be a rooted tree and $k \le k'$ be positive integers. Then
	\[ \tfrac k {k'} H_{k'}(T) - k \log \tfrac{k'}{k} \le H_k(T) \le H_{k'}(T). \]
\end{lemma}
\begin{proof}
	The second inequality holds by definition. For the first inequality, let $S' \subseteq V(T)$ be a node set with $H_{S'}(T) = H_{k'}(T)$. By \cref{p:Hm-ass-closed}, we can assume that $|S'| = k'$ and $S'$ is closed under taking descendants.
	Now construct $S$ by removing from $S'$ the $k'-k$ nodes with the largest rooted subtree size, i.e., each $u \in S'$ such that $|T_u|$ is among the $k'-k$ largest. Then, using \cref{p:Hm-desc-formula}, we have:
	\[ \sum_{u \in S'} \log \tfrac{k}{|T_u|} \le \tfrac{k'}{k} \sum_{u \in S} \log \tfrac{k}{|T_u|} = \tfrac{k'}{k} H_S(T). \]
	Clearly, $S$ is also closed under taking descendants. Again using \cref{p:Hm-desc-formula}, we have:
	\begin{align*}
		H_{k'}(T) = H_{S'}(T) = \sum_{u \in S'} \log \tfrac{k'}{|T_u|}
			\le \sum_{u \in S'} \log \left(\tfrac{k'}{k} \cdot \tfrac{k}{|T_u|}\right)
			& = k' \log \tfrac{k'}{k} + \sum_{u \in S'} \log \tfrac{k}{|T_u|} \\
			& \le k' \log \tfrac{k'}{k} + \tfrac{k'}{k} H_S(T).
	\end{align*}
	
	The first inequality of the lemma follows immediately.
\end{proof}

\subsection{Upper bound}\label{sec:short-ub}

We now prove a refinement of \cref{p:ub-partial}. It will be useful to fix the set $S$ of nodes that are \alg{cut} for the analysis. Note that the algorithm does not know $S$ in advance.

\begin{theorem}\label{p:ub-refined}
	Let $T^0$ be a rooted tree with $n$ nodes and $\ell$ leaves, let $p$ be a priority function for $T^0$, let $\sigma$ be a sequence of $m$ operations, and let $S$ be the set of nodes \alg{cut} by~$\sigma$. Then, the total time our data structure $X$ takes with this input is
	\[
		\Time(X, T^0, p, \sigma) \le \fO\left( m + n + \min(m, \ell) \log n + \sum_{v \in S} \log \tfrac{|S|}{|V(T^0_v) \cap S|} \right) \le \fO( H_m(T^0) + m + n ).
	\]
\end{theorem}

We will use the following weight function in our analysis. Fix $T^0$ and $S$, and let $n = |V(T^0)|$, $k = |S|$. For each node $v \in V(T^0)$, define
\[
	w'(v) = \begin{cases}
		1, & \text{ if } v \notin S\\
		\tfrac n k \cdot |V(T^0_v) \cap S|, & \text{ if } v \in S.
	\end{cases}
\]

Note that, in contrast to the weight function, $w$ defined in \cref{sec:ds:main-finish}, the function $w'$ can \emph{not} be computed in advance by the data structure, since $S$ is unknown at the start. Thus, it is crucial here that $w'$ not used explicitly by our SSM data structure.

For each super-node $x$, as before, define $w'(x)$ to be the maximum $w'(v)$ among all nodes $v$ contained in $x$. Let $W^*$ again be the maximum total weight of all super-nodes contained in a single structure $L_x$, at any point.

\begin{lemma}\label{p:ds-total-weights-refined}
	$W^* \le 2 n$.
\end{lemma}
\begin{proof}
	As in \cref{p:rt-total-weights}, $W^*$ is always a sum $\sum_{v \in Q} w'(v)$ for some subset $Q \subseteq V(T^0)$ of pairwise independent nodes. The nodes in $|Q \setminus S|$ trivially contribute at most $n$ to the sum. For the other nodes, we have:
	\[
		\sum_{v \in Q \cap S} w'(v) = \tfrac n k \cdot \sum_{v \in Q \cap S} |V(T^0_v) \cap S| \le \tfrac n k \cdot |S| = n,
	\]
	using that the sets $V(T^0_v)$ are pairwise disjoint.
	Thus, $W^* \le n + n = 2n$.
\end{proof}

Now most of the analysis transfers without much hassle.
From \cref{p:ds-total-weights-refined}, it follows that initialization takes $\fO(n)$ amortized time (cf.\ \cref{p:ds-init-time}). \Cref{p:ds-tree-min-time} does not depend on the weight function at all, so $\alg{tree-min}$ still takes constant amortized time. For \cref{p:ds-cut-time}, the only change is the final part of the analysis, where we perform a $\alg{cut}(v)$ such that $v$ contained in a super-leaf $y$, and the canonical super-forest operation is $\alg{split}(y)$. 
We determined that the running time is $\fO(\log \tfrac{W^*}{w'(y)})$, which we can now bound as follows. We have $w'(y) \ge w'(v)$ by definition, so $w'(y) \ge \tfrac n k \cdot |T^0_v \cap S|$. Thus
\[
	\log \frac{W^*}{w'(y)} \le \log \frac{2n}{\tfrac n k \cdot |V(T^0_v) \cap S|} = 1 + \log \frac{k}{|V(T^0_v) \cap S|},
\]
which implies:
\begin{lemma}\label{p:ds-cut-time-refined}
	Let $v$ be a node and let $y$ be the super-node containing $v$ in the current super-forest. $X.\alg{cut}(v)$ takes
	\begin{itemize}
		\item $\fO( 1 + \log \tfrac{k}{|V(T^0_v) \cap S|} )$ time if $y$ is a super-leaf and the canonical super-forest operation for $\alg{cut}(v)$ is $\alg{split}(y)$;
		\item $\fO(1)$ time if $y$ is a super-root; and
		\item $\fO(\log n)$ time otherwise.
	\end{itemize}
\end{lemma}

\begin{proof}[Proof of \cref{p:ub-refined}.]
	The proof essentially the same as for \cref{p:uopt-partial}, using \cref{p:ds-cut-time-refined}. The term $\min(m, \ell) \log n$ arises from the \alg{cut} operations involving non-leaf non-root super-nodes, of which there can be at most $\ell$ by \cref{p:num-cuts-splits}, and at most $m$ since there are $k \le m$ \alg{cut}s overall.
	This yields the $\fO( m + n + \min\{m, \ell\} \log n + H_S(T^0))$ bound.
	We have $H_S(T^0) \le H_k(T^0) \le H_m(T^0)$ by definition, and we have $n + \min\{m, \ell\} \log n \le \fO(H_m(T^0) + n)$ by \cref{p:lb-leaves-general}.
\end{proof}

\subsection{Lower bound}\label{sec:tts-lb}\label{sec:short-lb}

In this section, we show:
\begin{restatable}{theorem}{restateDTMFullLB}\label{p:dtm-full-lb}
	Let $X$ be a data structure for \pDTM{}. Then, for each rooted tree $T$ and integer $m \in \N_+$, there exists a priority function $p$ for $T$ and a sequence $\sigma$ of $m$ operations such that
	$\Time( X, T, p, \sigma ) \ge \Omega( H_m(T) + m + n )$.
\end{restatable}

Together with \cref{p:ub-refined}, this implies \cref{p:dtm-main}.
To prove \cref{p:dtm-full-lb}, we reduce from the following variant of the \pTS{} problem:

\begin{definition}\label{def:tts}
	In the \pTTS{} problem, we are given a rooted tree $T$, a $T$-monotone priority function $p$ on $V(T)$, and a number $k \le |V(T)|$. The task is to find and sort the $k$ nodes with lowest priority.
\end{definition}

A generalization where $T$ is an arbitrary DAG (i.e., corresponds to an arbitrary partial order) is called \problem{Top-$k$ DAG Sorting} and has been studied by Haeupler et al.~\cite{HaeuplerHladikEtAl2025a} as a variant of the \problem{DAG Sorting} problem. They give an optimal algorithm for the problem. Here, we are interested in giving a (more or less) explicit lower bound for the tree variant.

As usual, we use the information-theoretic lower bound. Let $P_T$ is the partial order on $V(T)$ that has Hasse diagram $T$. Denote by $\LEs_k(T)$ the set of length-$k$ sequences of distinct elements of $V(T)$ that appear as a prefix of some linear extension of $P_T$. Observe that $\LEs_k(T)$ is precisely the set of possible outputs of the \pTTS{} problem on $T$, so we need at least $\log |\LEs_k(T)|$ comparisons to solve it in the worst case. From this, we can obtain the following lower bound for \pDTM{}:

\begin{lemma}\label{p:lb-le-prefixes}
	Let $X$ be a data structure for \pDTM{}. Then, for each rooted tree $T$ and positive integer $k \le |V(T)|$, there exists a priority function $p$ for $T$ and a sequence $\sigma$ of $2k$ operations such that
	$\Time( X, T, p, \sigma ) \ge \log |\LEs_k(T)|$.
\end{lemma}
\begin{proof}
	Let $r$ be the root of $T$.
	The data structure
	$X$ can solve the \pTTS{} problem by $k$ times computing $x \gets X.\alg{tree-min}(r)$ and then calling $X.\alg{cut}(x)$. This computes the $k$ lowest-priority nodes of $V(T)$ in sorted order, as required. Thus, the information-theoretic lower bound $\log |\LEs_k(T)|$ applies.
\end{proof}

To lower bound $|\LEs_k(T)|$, we use another definition. Let $S \subseteq V(T)$, and let $\LEs_S(T)$ be the set of orderings of $S$ that are consistent with $P_T$, meaning that they are (possibly non-contiguous) subsequences of some linear extension of $P_T$.
Note that $\log |\LEs_S(T)|$ is a lower bound for a slightly different problem from \pTTS{}, where a fixed subset of nodes needs to be sorted. Nevertheless, we clearly have $|\LEs_k(T)| \ge |\LEs_S(T)|$ for every $S$ that is closed under taking descendants, a fact that we will use later.

\begin{lemma}\label{p:num-partial-tree-orders}
	For each rooted tree $T$ and $S \subseteq V(T)$, we have:
	\[ |\LEs_S(T)| = |S|! \cdot \prod_{u \in S} \frac{1}{|V(T_u) \cap S|}. \]
\end{lemma}
\begin{proof}
	We proceed by induction on $|T|$. By convention, when $S = \emptyset$, we let $\LEs_S(T)$ consists of a single, empty permutation, and we let $0! = 1$. Then, the lemma holds when $|S| = 0$. The second base case is when $|T| = |S| = 1$, then clearly $|\LEs_S(T)| = 1$.
	
	Now let $T$ be a rooted tree with root $r$. Let $c_1, c_2, \dots, c_k$ be the set of children of $T$, assuming $k > 0$. For each $i \in [k]$, let $T^i = T_{c_i}$ and $S_i = S \cap V(T^i)$.
	
	Suppose first that $r \notin S$. Then each permutation $\pi \in \LEs_S(T)$ is an arbitrary interleaving of the $k$ permutations $\pi_i \in \LEs_S(T^i)$. For each such set of permutations $\pi_1, \pi_2, \dots, \pi_k$, there are
	\[ |S|! \cdot \prod_{i=1}^{k} \frac{1}{|S_i|!} \]
	many such interleavings. Thus, we have:
	\begin{align*}
		|\LEs_S(T)| & = |S|! \cdot \prod_{i=1}^{k} \frac{1}{|S_i|!} \cdot \prod_{i=1}^k |\LEs_S(T^i)|\\
		& = |S|! \cdot \prod_{i=1}^{k} \frac{1}{|S_i|!} \cdot \prod_{i=1}^k \left( |S_i|! \cdot \prod_{u \in S_i} \frac{1}{|V(T_u) \cap S_i|} \right) \tag{induction}\\
		& = |S|! \cdot \prod_{u \in S} \frac{1}{|V(T_u) \cap S|},
	\end{align*}
	where the last step uses the fact that $V(T_u) \cap S_i = V(T_u) \cap S$ if $u \subseteq S_i$.
	
	If $r \in S$, then each permutation $\pi \in \LEs_S(T)$ ends with $r$, and removing $r$ again yields an interleaving of the $k$ permutations $\pi_i \in \LEs_S(T^i)$. The calculation above then yields the formula
	\begin{align}
		|\LEs_S(T)| = |S \setminus \{r\}|! \cdot \prod_{u \in S \setminus \{r\}} \frac{1}{|V(T_u) \cap S|}.\label{eq:hook-len-root}
	\end{align}
	Observe that $|V(T_r) \cap S| = |S|$. Thus, (\ref{eq:hook-len-root}) is equal to the claimed formula.
\end{proof}

We can now finally prove:

\restateDTMFullLB*
\begin{proof}
	We first show a $\Omega(n)$ lower bound. Indeed, observe that if the first operation is a \alg{tree-min}, then we need $n-1$ comparisons for it.
	Observe that a $\Omega(m)$ lower bound also trivially holds.

	We now show a $\Omega(H_m(T) - m)$ lower bound, which completes the proof.
	Let $k = \lfloor \tfrac m 2 \rfloor$. \Cref{p:lb-le-prefixes} implies $\Time( X, T, p, \sigma ) \ge \log |\LEs_k(T)|$. By \cref{p:Hm-ass-closed}, there is a descendant-closed set $S$ of $k$ nodes with $H_k(T) = H_S(T)$. Now, since $|\LEs_k(T)| \ge |\LEs_S(T)|$ by definition, we have:
	\begin{align*}
		\Time( X, T, p, \sigma ) & \ge \log |\LEs_S(T)|
			= \log \left(|S|! \cdot \prod_{u \in S} \frac{1}{|V(T_u) \cap S|}\right)\tag{\cref{p:num-partial-tree-orders}}\\
		& \ge \left( \sum_{u \in S} \frac{k}{|V(T_u) \cap S|} \right) - k \log e
			\ge H_S(T) - \fO(k) = H_k(T) - \fO(m).\tag*{\qedhere}
	\end{align*}
\end{proof}

\section{Cartesian trees and path minima}\label{sec:ct-pm}

In this section, we show how to compute Cartesian trees and use them to implement a \pPM{} data structure. We start with the case where the input graph $G$ is a tree (with a priority function $p$).

Recall DTM-based algorithm described in the introduction. Take an arbitrary rooting $T$ of $G$. Initialize the DTM data structure with $(T,p)$. Find the overall minimum $v$ (with an arbitrary \alg{tree-min}), and then delete all edges incident to $v$.
We can now identify each remaining component $C$ by the unique neighbor of $v$ contained in $C$.
Recursively compute a Cartesian tree on each such component and attach it to $v$. This clearly yields the Cartesian tree $\Ind(G,p)$.

Observe that we perform no more than $2n$ operations, so, using our DTM data structure, the total running time is $\fO(n + H(T)) = \fO(H(T))$. We have shown:

\begin{theorem}\label{p:ct-algo-trees}
	We can solve the \pCT{} problem on trees in $\fO(H(T))$ time, where $T$ is some rooting of the input tree.
\end{theorem}

In \cref{sec:et-lb}, we show a matching lower bound, thereby proving universal optimality. In \cref{sec:ct-graphs}, we extend the algorithm and lower bound to general graphs, and in \cref{sec:pm}, we discuss the \pPM{} problem.

\subsection{The number of elimination trees}\label{sec:et-lb}

Observe that every elimination tree $T$ on a given graph $G$ is the Cartesian tree induced by some priority function $p$ for $G$. This implies that every elimination tree of $G$ may appear as the output of the \pCT{} problem. Thus, for any fixed graph, every \pCT{} algorithm requires $\log |\ETs(G)|$ comparisons in the worst case by the information-theoretic bound. In this section, we estimate $|\ETs(G)|$ in the case when $G$ is a tree.\footnote{We could directly reduce from \pTS{} (\cref{def:tree-sorting}) instead. However, explicitly giving an information-theoretic lower bound is easier to generalize to arbitrary graphs.}

We need the following definition.
An elimination tree is called \emph{degenerate} if it has precisely one leaf, or, equivalently, if all inner nodes have degree one. Also recall the following definitions from \cref{sec:lb}:
Given a rooted tree $T$, a priority function $p$ is \emph{$T$-monotone} if for each node $v$ and each ancestor $u$ of~$v$ in $T$, we have $p(v) < p(u)$, and $\LEs(T)$ denotes the set of ascending orderings on $V(T)$ induced by $T$-monotone priority functions.

\begin{lemma}\label{p:mon-deg}
	Let $G$ be a tree, let $T$ be an arbitrary rooting of $G$, and let $p$ be a $T$-monotone priority function for $G$. Then $\Ind(G,p)$ is degenerate, and reading the nodes of $\Ind(G,p)$ from root to leaf yields the ordering induced by $p$.
\end{lemma}
\begin{proof}
	Intuitively, constructing $\Ind(G,p)$ in a top-down manner will progressively remove leafs of $T$ (which are also leaves of $G$). The remaining graph stays connected, so every inner node has one child.
	
	For a more formal proof, let $S = \Ind(G,p)$, and suppose some node $v \in V(S)$ has two or more children. By definition, the induced subgraph $T[V(S_v)]$ is connected, and by assumption, $T[V(S_v)] \gminus v$ is not connected. Then $v$ cannot be a leaf of $T[V(S_v)]$. But then $v$ has a child $u$ in $T[V(S_v)]$, and $p(u) < p(v)$ by assumption. Therefore, $v$ is not the minimum-priority node in $T[V(S_v)]$, a contradiction.
	
	The second statement of the lemma is directly implied by the fact that for every node $u$ and its parent $v$ in $\Ind(G,p)$, we have $p(v) < p(u)$.
\end{proof}

\begin{corollary}\label{p:diff-mon-deg}
	Let $G$ ba a tree, let $T$ be a rooting of $G$, and let $p, p'$ be $T$-monotone priority functions for $G$ that induce different orders on $V(G)$. Then $\Ind(G,p) \neq \Ind(G,p')$.
\end{corollary}

\Cref{p:diff-mon-deg} immediately implies that the number $|\ETs(G)|$ of elimination trees is at least the number of orderings of $V(G)$ induced by $T$-monotone priority functions. By definition, the latter set is precisely $\LEs(T)$. Hence, we have:

\begin{lemma}\label{p:lb-ct-trees}
	Let $G$ be a tree and $T$ be an arbitrary rooting of $G$.
	Then $|\ETs(G)| \ge |\LEs(T)|$, so solving \pCT{} on a graph $G$ requires at least $\log |\LEs(T)|$ comparisons.
\end{lemma}

Recall that $\log \LEs(T) \ge H(T) - \fO(n)$ (\cref{p:tree-les}). To show optimality of \cref{p:ct-algo-trees}, we additionally need the following lower bound, which we can already state for general graphs:

\begin{lemma}\label{p:et-lb-easy}
	Let $G$ be a connected graph on $n$ nodes. Then $|\ETs(G)| \ge 2^{n-1}$, implying that $\log |\ETs(G)| \ge n - 1$.
\end{lemma}
\begin{proof}
	Each connected graph with two or more vertices has at least two vertices $u, v$ such that $G - u$ and $G-v$ are both connected (e.g., take two leaves of a spanning tree). Thus, when constructing a degenerate elimination tree, there are always two vertices that can be chosen to be the root, except when $|V(G)| = 1$.
\end{proof}

Thus, combining \cref{p:lb-ct-trees,p:et-lb-easy}, we obtain a lower bound of $\max \{ H(T) - \fO(n), n - 1 \} \ge \Omega( H(T) )$. This means that \cref{p:ct-algo-trees} is optimal.

\subsection{Computing Cartesian trees on graphs}\label{sec:ct-graphs}

In this section, we show how to reduce the general \pCT{} problem to the tree case, in universally optimal time. Ultimately, we will prove:

\restateCartTree*

The following lemma is the main insight needed for the reduction.

\begin{lemma}[Zhu and Mutchler \cite{ZhuMutchler1994}]\label{p:min-edge}
	Let $(G,p)$ be a graph with vertex priorities.
	Let $C$ be a cycle and let $e$ be an edge incident to the minimum-priority vertex of~$C$. Then $\Ind(G,p) = \Ind(G \gminus e,p)$.
\end{lemma}
\begin{proof}
	Let $v$ be the minimum-priority vertex of $G$, and let $G' = G \gminus e$.
	If $v$ is contained in $C$, then $e$ is incident to $v$ by definition. Thus $G \gminus v = G' \gminus v$. The Cartesian trees $\Ind(G,p)$ and $\Ind(G',p)$ both have root $v$, and the child subtrees of $v$ are Cartesian trees on the components of $G \gminus v = G' \gminus v$. Thus, we have $\Ind(G,p) = \Ind(G',p)$
	
	Now suppose $v$ is not contained in $C$. Then $C$ is fully contained in some connected component $H$ of $G \gminus v$. By induction, $\Ind(H,p) = \Ind(H \gminus e, p)$. This implies that also $\Ind(G,p) = \Ind(G',p)$.
\end{proof}

\Cref{p:min-edge} was previously used by Zhu and Mutchler~\cite{ZhuMutchler1994} to compute Cartesian trees on arbitrary graphs in the setting where priorities are already pre-sorted. Interestingly, \cref{p:min-edge} is also useful to lower bound $|\ETs(G)|$, as we show now.

\begin{lemma}\label{p:reduce-deg-et}
	Let $G$ be a graph, let $H$ be a spanning tree of $G$, and let $T$ be a rooting of~$H$. Then, for every $T$-monotone priority function for $G$, we have $\Ind(G,p) = \Ind(H,p)$.
\end{lemma}
\begin{proof}
	We show how to transform $G$ to $H$ by progressively removing edges with \cref{p:min-edge}. Suppose $G \neq H$, then $G$ contains a cycle $C$. Let $v$ be the minimum-priority vertex in~$C$. Then no child of $v$ in $T$ (if any exists) can be contained in $C$; otherwise, that child would have lower priority (since $p$ is $T$-monotone). Thus, at most one edge incident to $v$ in $C$ can lead to the parent of $v$ in $T$, so the other edge $e$ is not contained in $H$. We can remove $e$; by \cref{p:min-edge}, we have $\Ind(G,p) = \Ind(G-e,p)$. Repeating this will eventually transform $G$ into $H$.
\end{proof}

Using \cref{p:reduce-deg-et,p:diff-mon-deg,p:et-lb-easy}, we obtain our lower bound:

\begin{lemma}\label{p:et-lb-general}
	Let $G$ be a graph, and let $T$ be a rooting of a spanning tree of $G$. Then $\log |\ETs(G)| \ge \Omega(H(T))$.
\end{lemma}

We now show how to generalize our \pCT{} algorithm to arbitrary graphs.
For a graph $(G,p)$ with vertex priorities, define $\wmin_p \colon E(G) \rightarrow \Prios$ as follows. For each edge $e = \{u,v\}$, let $\wmin_p(e) = \min( p(u), p(v) )$. From \cref{p:min-edge}, we easily get:

\begin{corollary}[Zhu and Mutchler \cite{ZhuMutchler1994}]\label{p:max-sptree}
	Let $(G,p)$ be a graph with vertex priorities. Let $H$ be the \emph{maximum} spanning tree on $(G,\wmin_p)$. Then $\Ind(G,p) = \Ind(H,p)$.
\end{corollary}

Computing the maximum spanning tree is equivalent to computing the minimum spanning tree in the comparison model. Thus, using Chazelle's algorithm~\cite{Chazelle2000}, we can compute $H$ in time $\fO(m \cdot \alpha(m,n))$, where $m = |E(G)|$, $n = |V(G)|$, and $\alpha$ is an inverse of the Ackermann function. Since $\Omega(m)$ is an obvious lower bound, this gives us an \emph{almost} universally optimal algorithm. However, we can modify a different algorithm~\cite{HaeuplerHladikEtAl2024a} to keep the running time within our budget.

\begin{restatable}{lemma}{restateJarnikPrim}\label{p:jp}
	There exists an algorithm that computes a maximum spanning tree of an edge-weighted graph $(G,w)$ in time $\fO( |E(G)| + |\ETs(G)| )$.
\end{restatable}
\begin{proof}[Proof sketch.]
	We use a variant of the algorithm of Haeupler, Hladík, Rozhoň, Tarjan, and Tětek~\cite{HaeuplerHladikEtAl2024a} for Single-Source-Shortest-Paths. Their algorithm is Dijsktra's algorithm with a special heap, called a \emph{working-set fibonacci-like heap}. We can use this heap in the classical \emph{Dijkstra-Jarník-Prim} algorithm~\cite{Jarnik1930,Prim1957,Dijkstra1959} for minimum spanning trees.
	The analysis is the same, and with some simple observations, we can bound the running time by $\fO( \log |\ETs(G)| )$.
	A full technical proof is found in \cref{app:jp}, using the simplified analysis of Hoog, Rotenberg, and Rutschmann~\cite{HoogRotenbergEtAl2025}.
\end{proof}

We are now read to prove:
\restateCartTree*
\begin{proof}
	First, compute the edge priority function $w^{\min}_p$. Then, compute a maximum spanning tree $S$ of $(G,w^{\min}_p)$ with \cref{p:jp}. Finally, compute $\Ind(S,p)$ using \cref{p:ct-algo-trees}. By \cref{p:max-sptree}, we have $\Ind(S,p) = \Ind(G,p)$.
	
	Let $n = |V(G)|$ and $m = |E(G)|$.
	The running time of computing $w^{\min}_p$ is $\fO(m+n)$, and the running time to compute $S$ is $\fO( m + |\ETs(G)|)$. Computing $\Ind(S,p)$ takes $\fO(H(T))$ time, where $T$ is some rooting of $S$. By \cref{p:et-lb-general}, we have $H(T) \le \fO(\log |\ETs(G))$, so the running time is $\fO(m + \log |\ETs(G))$, as desired. This is universally optimal by the information-theoretic lower bound $\Omega(\log |\ETs(G)|)$ and the trivial lower bound $\Omega(m)$.
\end{proof}

\subsection{Path Minima}\label{sec:pm}

We now use Cartesian trees to build a data structure for the \pPM{} problem. We follow common approaches for the edge-weighted variant of the problem~\cite{Chazelle1987,DemaineLandauEtAl2014}.¸
The following lemma is crucial.

\begin{lemma}
	Let $G$ be a tree with priority function $p$, and let $T = \Ind(G,p)$. Then, for every pair $u, v \in V(G)$, the minimum-priority vertex on the path between $u$ and $v$ is precisely the lowest common ancestor of $u$ and $v$ in $T$.
\end{lemma}
\begin{proof}
	Consider the lower common ancestor $a$ of $u$ and $v$, let $H = G[V(T_a)]$, and let $P$ be the path between $u$ and $v$ in $G$. Then $H$ is connected by definition, but $H \gminus a$ has at least two connected components, one containing $u$, and one containing $v$. Now observe that:
	\begin{itemize}
		\itemsep0pt
		\item Since $G[V(T_a)]$ is connected, we have $V(P) \subseteq V(H)$.
		\item Since removing $a$ from $H$ disconnects $u$ and $v$, we have $a \in V(P)$.
		\item By definition, $a$ is the minimum-priority vertex in $H$.
	\end{itemize}
	These three observations imply that $a$ is the minimum-priority vertex of $P$.
\end{proof}

We can now solve \pPM{} as follows: Compute a Cartesian tree $T$ on $(G,p)$, and then preprocess $T$ such that lowest common ancestors in $T$ can be found in constant time. This is possible in time $\fO(|V(T)|) \le \fO(\log |\ETs(G)|)$ using any one of a variety of data structures~\cite{HarelTarjan1984,SchieberVishkin1988,BenderFarach-Colton2000,FischerHeun2006}.
Thus, we have:

\restatePMDS*

We now prove a matching lower bound. The idea is that we can solve \pCT{} solely with a (static!) \pPM{} data structure. This means that the overall number of comparisons made by the \pPM{} data structure must be at least $\log |\ETs(G)|$. Since the number of $\alg{path-min}$ queries is large, we do not get any nontrivial lower bound for the running time. However, if we assume that no comparisons are allowed within $\alg{path-min}$ queries, then preprocessing needs $\Omega(\log |\ETs(G)|)$ comparisons.

\begin{lemma}
	Let $D$ be a \pPM{} data structure such that \alg{path-min} does not make any comparisons (but may take arbitrary time). Then, initializing $D$ with a graph $G$ requires $\Omega( \log |\ETs(G)| )$ time.
\end{lemma}
\begin{proof}
	As outline above, we show how to compute $\Ind(G,p)$ using $D$, with zero additional comparisons (but arbitrary running time). Observe that is suffices to show how to find the minimum-priority vertex of any connected subgraph $H$ of $G$.
	
	To do this, first enumerate the leafs $u_1, u_2, \dots, u_k$ of $H$, in any order. Then arbitrarily choose $v_0 \in V(H)$. For each $i \in [k]$, compute $v_i = D.\alg{path-min}(u_i, v_{i-1})$. Then $v_k$ is the minimum-priority vertex of $H$, since the \alg{path-min} queries cover the whole tree $G$.
\end{proof}

Recall that our data structure from \cref{p:pmds} performs no comparisons in \alg{path-min} queries. Thus, we have shown:

\restatePMDSUOpt*

Note that we do not achieve universal optimality even if \emph{one} comparison in each \alg{path-min} query is allowed. If the input graph $G$ is a star, then it suffices to compare every leaf with the central vertex of $G$ in preprocessing. Afterwards, only a single comparison is needed in each query.
On the other hand, for general trees, a $\Omega(n \log\log n)$ preprocessing lower bound is known~\cite{Pettie2006}.\footnote{A one-comparison \alg{path-min} data structure would imply a two-comparisons \emph{online MST verification} data structure, which requires $\Omega(n \log \log n)$ preprocesssing~\cite{Pettie2006}.}
Finding a universally optimal preprocessing algorithm for one-comparison \alg{path-min} queries is an interesting open problem.

\paragraph{Bottleneck vertex queries.}
Recall the following definition from the introduction. Let $G$ be an arbitrary graph and $p$ be a priority function for $G$. The \emph{bottleneck} of a path in $G$ is the minimum-priority vertex on that path, and a $\alg{bottleneck}(u,v)$ query returns the maximum bottleneck among all paths between $u$ and $v$ in $G$.

Observe that if $G$ is a tree, then \alg{bottleneck} queries are exactly \alg{path-min} queries. As observed previously~\cite{Hu1961,ShapiraYusterEtAl2011,DemaineLandauEtAl2014}, the general case reduces to the tree case, in essentially the same way as for the \pCT{} problem.

\begin{lemma}[Shapira, Yuster, Zwick~\cite{ShapiraYusterEtAl2011}]\label{p:min-edge-bottleneck}
	Let $(G,p)$ be a graph with vertex priorities.
	Let $C$ be a cycle and let $e$ be an edge incident to the minimum-priority vertex of~$C$.
	Then, for every pair $u,v \in V(G)$, the query $\alg{bottleneck}(u,v)$ returns the same answer for $(G,p)$ and $(G \gminus e, p)$.
\end{lemma}
\begin{proof}[Proof sketch.]
	Consider a path containing the edge $e$. Then, we can replace $e$ by $C \setminus \{e\}$. This cannot change the bottleneck of the path, since the minimum-priority vertex of $C$ is an endpoint of $e$, and was thus already contained in the path.
\end{proof}

As in \cref{sec:ct-graphs}, \cref{p:min-edge-bottleneck} implies that we can first compute a minimum spanning tree $H$ on $(G,\wmin_p)$ using the Jarník-Prim Dijkstra algorithm (\cref{p:jp}), and then use our \pPM{} data structure.
\Cref{p:min-edge-bottleneck} also implies that \alg{bottleneck} queries on $G$ are not easier than \alg{path-min} queries on $H$; thus, our data structure is universally optimal.

\section{Semigroup sums}\label{sec:semigroups}

In this section, we show:

\restateDSTSMain*

Recall that the input is now a a vertex-weighted tree $(T,w)$, with weights from some commutative semigroup $(S,+,0)$, and we need to answer the query $\alg{tree-sum}(v)$, which returns the sum of all nodes in the same tree as $v$.

Most of our data structure can be used as-is, with the only change being changing $\min(x,y)$ to $x+y$. We list all places where our data structure computes minima in the following:
\begin{itemize}
	\item The dynamic tree $D$ maintains the minimum super-node.
	We can easily maintain the sum of all super-node weights instead.
	
	\item In our splay-tree-based \pSSM{} data structure, each node $x$ in the splay tree $T$ maintains the minimum priority $q_x$ of $T_x$. We can instead maintain the sum of weights in $T_x$, and easily update it for each relevant binary tree operation (rotation, attaching to the root, detaching from the root).
	The crucial observation here is that a sum $q_x$ can be recomputed from the weight of $x$ and the sums of its (up to two) children.
	
	\item In contrast, the path-DTM data structures for each chain depend on certain properties of the minimum function and need to be entirely replaced; see below.
	
	\item Finally, to answer $\alg{tree-min}$ queries, we compute the minimum of three values; instead, we can simply take their sum.
\end{itemize}

It remains to show how to solve \pDSTS{} on a path. As mentioned in the introduction, a sufficiently fast \emph{semigroup range query} data structure does not exist. However, we can instead directly solve the problem, with linear preprocessing time and constant query time.

Our data structure is based on an algorithm of Weiss~\cite{Weiss1994} for constructing \emph{Cartesian trees} on paths in linear time. While a much simpler algorithm for this task exists~\cite{GabowBentleyEtAl1984}, Weiss's approach is still useful to us.

\begin{lemma}
	There is a data structure for the \pDSTS{} problem on paths with linear preprocessing time and constant query time.
\end{lemma}
\begin{proof}[Proof sketch.]
	Weiss's algorithm for Cartesian trees works as follows: Identify the minimum vertex of the path, make it the root of the Cartesian tree, remove it, and then recurse on the (up to) two remaining paths.
	It is easy to check that, actually, the algorithm supports removing an \emph{arbitrary edge} instead of specifically the (up to) two edges incident to the minimum vertex. The minimum of each path is always maintained explicitly. This gives us a data structure for DTM.
	
	Generalizing this data structure to semigroup turns out to be quite easy. Suppose the input path has $n = 2^k$ nodes.
	There are two places where minima are computed:
	\begin{itemize}\itemsep0pt
		\item At the start, for each interval of the form $I_{p,q} = [q \cdot 2^p+1, (q+1) \cdot 2^p]$ with $0 \le p \le k$ and $0 \le q < n / 2^p$, the algorithm computes the minimum of the subpath corresponding to $I$, in $\fO(n)$ total time. This can easily be done for the respective semigroup sums.
		
		\item For each current subpath, the data structure stores precisely the minimum set of intervals $I_{p,q}$ whose union corresponds to the subpath. These intervals are stored in a double-ended queue (deque) that supports finding the minimum in addition to the usual deque operations, all in constant time.
		
		This deque is due to Gajewska and Tarjan~\cite{GajewskaTarjan1986}, and is called, somewhat deceptively, a \emph{deque with heap order}. This data structure implements a deque with the familiar technique of using two stacks for the left and right part of the deque. An amortization argument shows that $m$ deque operations can be performed using $\fO(m)$ stack operations.
		
		Each stack maintains the minimum of each of its suffixes. This requires $\fO(1)$ time to maintain per stack operation, and allows computing the minimum of the deque by taking the minimum of the two whole stacks.
		Observe that the same can be done with semigroup sums.\qedhere
	\end{itemize}
\end{proof}

With this, we have shown:

\restateDSTSMain*

\paragraph{Remark.} The \pPM{} problem can also be generalized to semigroups~\cite{Tarjan1979,AlonSchieber1987}. However, our Cartesian-tree based approach does not work for this generalization.

\section{Edge-weighted graphs}\label{sec:edge-weights}

\newcommand{\tG}{\widetilde{G}}
\newcommand{\tT}{\widetilde{T}}
\newcommand{\tp}{\widetilde{p}}
\newcommand{\tS}{\widetilde{S}}
\newcommand{\tk}{\widetilde{k}}

In this section, we consider variants of our problems where priorities are given for edges instead of vertices. The following construction is central.

Let $G$ be a graph, and let $p \colon E(G) \rightarrow \Prios$ be a priority function on the \emph{edges} of $G$. The graph $\widetilde{G}$ is constructed as follows (see \cref{fig:edge-reduct}): Replace each edge $e$ between vertices $u, u'$ by a vertex $v_e$, which is adjacent to $u$ and $u'$. I.e., the edge $e$ is \emph{subdivided} with a vertex $v_e$. Define $\widetilde{p}$ as $\widetilde{p}(u) = \infty$ for $u \in V(G)$, and $\tp(v_e) = p(e)$ for $e \in E(G)$.\footnote{Note that several vertices can have the same priority $\widetilde{p}(u) = \infty$; this will be addressed later when necessary.}

\begin{figure}
	\centering
	\includegraphics{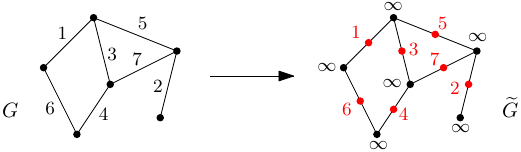}
	\caption{The reduction from edge-weighted to vertex-weighted graphs. Priorities are indicated by numbers, new vertices $v_e$ are colored red.}\label{fig:edge-reduct}
\end{figure}

For a rooted tree $T$ with a priority function $p$ on its edges, we construct $\tT$ similarly: Replace each edge $e$ from child $u$ to parent $u'$ with a new node $v_e$, and make $u$ a child of $v_e$ and $u'$ the parent of $v_e$. We also write $v_u = v_e$.
Define $\tp$ as above. Clearly, if $T$ is a rooting of a tree $G$, then $\tT$ is a rooting of $\tG$.

We now show that the transformation from $T$ to $\tT$ does not increase the tree entropy by much.

\begin{lemma}\label{p:Hm-edge-bound}
	Let $T$ be a rooted tree on $n$ nodes, and let $k \in \N_+$. Then $H_{k}(\tT) \le 2 H_k(T) + 2k$.
\end{lemma}
\begin{proof}
	By \cref{p:Hm-ass-closed}, there exists a set $\tS \subseteq V(\tT)$ of size $k$ that is closed under taking descendants and satisfies $H_k(\tT) = H_{\tS}(\tT)$. Let $S = \tS \cap V(T)$. Since $|S| \le k$, we have $H_k(T) \ge H_S(T)$ by definition. Thus, it suffices to show that $H_{\tS}(\tT) \le 2 H_S(T) + 2|S|$.
	
	Observe that $S$ is also closed under taking descendants.
	Since each $u \in \tS \setminus V(T)$ has a single child in $T$, and that child is contained in $S$, we have $|\tS| \le 2 |S|$.
	
	We now estimate $H_{\tS}(\tT)$.
	For each node $u \in S$, we have $|\tT_u| = 2|T_u| - 1$. Further, for each non-root node $u \in S$, we have $|\tT_{v_u}| > |\tT_u|$. Using \cref{p:Hm-desc-formula} and the fact that $S$ and $\tS$ are both closed under taking descendants, we have:
	\begin{align*}
		H_{\tS}(\tT) = \sum_{u \in \tS} \log \frac{|\tS|}{|\tT_u|}
			< 2 \sum_{u \in S} \log \frac{|\tS|}{|\tT_u|}
			\le 2 \sum_{u \in S} \log \frac{2|S|}{2|T_u|-1}\\
			\le 2 \sum_{u \in S} \log \frac{2|S|}{|T_u|}
			\le 2|S| + 2 \sum_{u \in S} \log \frac{|S|}{|T_u|}
			\le 2|S| + 2 H_S(T).\tag*{\qedhere}
	\end{align*}
\end{proof}

\subsection{Decremental tree minima with edge weights}\label{sec:edtm}
\newcommand{\pEDTM}{\problem{Edge Decremental Tree Minima}}

Consider the following variant of \pDTM{}, which we call \pEDTM{} (EDTM). Given is an initial rooted tree $T$ and a priority function $p$ on the \emph{edge set} $E(T)$. $\alg{tree-min}(v)$ now reports the minimum-priority edge in the tree containing $v$, or $\infty$ if that tree contains no edges. \alg{cut} works as usual.

We can now use our DTM data structure to solve the EDTM problem as follows. Consider an EDTM input $(T,p)$. First, build the corresponding DTM input $(\tT, \tp)$, and initialize a DTM data structure on it. Whenever an edge $e$ is deleted in $T$, we delete the two edges adjacent to the corresponding vertex $v_e$ in $T$. More formally, for a given $\alg{cut}(u)$ in $T$, we perform $\alg{cut}(u)$ and $\alg{cut}(v_u)$ in $\tT$.
\alg{tree-min} queries can be answered directly.

It is easy to see that this is a correct EDTM data structure. From \cref{p:ub-refined}, we obtain:
\begin{lemma}\label{p:edtm-ub}
	Let $T$ be a rooted tree with $n$ nodes, let $p$ be a priority function for $E(T)$, and let $\sigma$ be a sequence of $m$ EDTM operation applied to $(T,p)$.
	Then, our data structure takes $\fO( H_{2m}(\tT) + m + n )$ total time for $(T,p,\sigma)$.
\end{lemma}

We now prove a corresponding lower bound.
We show that the \pTTS{} problem on $T$ can be solved with an EDTM data structure on $T$, using $2k$ operations. Given is a tree $T$ and a $T$-monotone vertex priority function $p$. First, compute an edge priority function $p'$, where $p'(e) = p(u)$ for each edge $e$ that connects a child $u$ with its parent. Now build an EDTM data structure on $(T,p')$. Let $r$ be the root of $r$. Call $\alg{find-min}(r)$, which reports the minimum-priority edge $e$. Output $p'(e)$, delete $e$, and repeat until there are no more edges. At the end, output the root priority $p(r)$.

Observe that whenever this algorithm finds an edge $e$, one of its endpoints must currently be a leaf (since $p$ is $T$-monotone). Thus, all remaining edges stay in the same component as the initial root $r$.
The output of the algorithm is clearly the priorities of $(T,p)$ in sorted order. From our lower bound for \pTTS{} (\cref{sec:tts-lb}), we obtain:

\begin{lemma}\label{p:edtm-lb}
	For every EDTM data structure $X$, every tree $T$, and every positive integer $k \le |V(T)|-1$, there exists a sequence $\sigma$ of $2k$ operations such that $X$ takes total time $\log |\LEs_k(T)| \ge \Omega(H_k(T) - k)$ for $(T,p,\sigma)$.
\end{lemma}

We now show universal optimality. Let $T$ be a rooted tree, and let $m \in \N_+$ be the number of operations. By \cref{p:edtm-ub}, the running time of our algorithm for any edge priority function $p$ and any operation sequence of length $m$ is $\fO(H_{2m}(\tT) + m + n )$, which is $\fO( H_m(T) + m + n )$ by \cref{p:Hm-mon,p:Hm-edge-bound}. Considering the trivial lower bound $\Omega(m+n)$, this matches \cref{p:edtm-lb}.

\begin{theorem}
	Our EDTM data structure is universally optimal.
\end{theorem}

\subsection{Cartesian edge-partition trees}
Recall the definition of \emph{Cartesian edge-partition trees} (Cartesian EPTs) from the introduction. A Cartesian EPT on a connected graph $G$ with edge priority function $p$ is a rooted binary tree $T$ with $V(T) = E(G) \cup V(G)$, where the edges of $G$ correspond to inner nodes of $T$, and vertices of $G$ correspond to leaves of $T$. If $G$ consists of a single vertex, then $T$ also consists of that vertex. Otherwise, the root of the EPT is the minimum-priority edge of $(G,p)$, and the the up to two child subtrees are recursively constructed on the up to two components of $G \gminus e$.
We also write $\IndE(G,p) = T$.
Let $\EPTs(G)$ denote the set of possible Cartesian EPTs on $G$.

It is known that Cartesian EPTs on \emph{trees} can be constructed in $\fO( n \log n )$ time~\cite{Chazelle1987}, or in $\fO( n \log \ell )$ time if $G$ has $\ell$ leaves~\cite{DeanMohan2013}. In this section, we give a universally optimal algorithm for arbitrary graphs.

Fix some connected graph $G$ and edge priority function $p$ on $G$. Recall that $\tp$ assigns multiple vertices the same priority $\infty$. Since this does not play well with our definition of Cartesian trees, we introduce the following tie-breaking.
Fix an arbitrary ordering $\pi$ on $V(G)$, and redefine $\tp(v) = \infty_v$; where $\infty_v$ is a priority that is larger than all edge priorities, and $\infty_v < \infty_u$ if $v$ precedes $u$ in $\pi$. We can now prove the following.

\begin{lemma}\label{p:epts-reduct}
	Let $G$ be a graph with edge priority function $p$. Then $\IndE(G,p)$ is isomorphic to $\Ind(\tG,\tp)$ under the mapping $v \rightarrow v$ for $v \in V(G)$ and $e \rightarrow v_e$ for $e \in E(G)$.
\end{lemma}
\begin{proof}
	If $G$ contains only one vertex, then both $\IndE(G,p)$ and $\Ind(\tG,\tp)$ consist of that vertex only.
	
	Otherwise, let $e$ be the minimum-priority edge of $G$, and recall that $v_e$ must be the minimum-priority vertex of $\tG$.
	
	Suppose first that $G \gminus e$ is disconnected, so it has two components $H_1$, $H_2$.
	Then $\tG \gminus v_e$ also has two components $H'_1$, $H'_2$. It is easy to see that $H'_1 = \widetilde{H}_1$ and $H'_2 = \widetilde{H}_2$ (for some assignment of the names $H_1, H_2$).
	By induction, $\Ind(H_i,p) = \Ind(\widetilde{H}_i,\tp)$ for $i \in \{1,2\}$. Since $\IndE(G,p)$ consists of the root $e$ with children $\Ind(H_1,p)$, $\Ind(H_2,p)$, and $\Ind(\tG,\tp)$ consists of the root $v_e$ with children $\Ind(\widetilde{H}_1,\tp)$, $\Ind(\widetilde{H}_2,\tp)$, our claim follows.
	
	Suppose now $H \coloneq G \gminus e$ is connected. Clearly, $\tG \gminus v_e = \widetilde{H}$. Our claim follows similarly.
\end{proof}

Observe that \cref{p:epts-reduct} in particular implies that the choice of $\pi$ for tie-breaking does not affect the Cartesian tree $\Ind(\tG,\tp)$.

We now give the universally optimal algorithm. Given is a graph $G$ with $n$ vertices and $m$ edges, and an edge priority function $p$. Start by computing $(\tG,\tp)$ in $\fO( m + n )$ time. Then, compute $\Ind(\tG,\tp)$ with \cref{p:univ-opt-et}. Using \cref{p:epts-reduct}, we can map the result to the correct EPT $\IndE(G,p)$.

The total running time is $\fO( m + n + \log |\ETs(\tG)| )$. Since \cref{p:epts-reduct} defines a bijection between $\EPTs(G)$ and $\ETs(\tG)$, we have $|\EPTs(G)| = |\ETs(\tG)|$.

\begin{theorem}\label{p:ct-graphs-ub}
	Given a connected graph $G$ with edge priority function $p$, we can compute $\IndE(G,p)$ in time $\fO( |E(G)| + \log |\EPTs(G)|$.
\end{theorem}

By the information-theoretic lower bound and the trivial $\Omega(|E(G)|)$ lower bound, \cref{p:ct-graphs-ub} is universally optimal.

\subsection{Edge-weighted path-minimum and bottleneck queries.}
It is well known that Cartesian EPTs can be used to answer edge-weighted variants of \alg{path-min} and \alg{bottleneck} queries~\cite{Chazelle1987,DemaineLandauEtAl2014}, much like Cartesian trees can be used for the vertex-weighted variants. For \alg{bottleneck} queries on general graphs, again a spanning tree reduction is used~\cite{Hu1961,ShapiraYusterEtAl2011,DemaineLandauEtAl2014}. Using the same techniques as above, we obtain a universally optimal algorithm to preprocess a graph $G$ with edge priorities such that \alg{bottleneck} queries can be answered without further comparisons. We omit the details.

\paragraph{Acknowledgments.} The author would like to thank Tomasz Kociumaka and Marek Sokołowski for useful suggestions.

\small
\bibliographystyle{alphaurl}
\bibliography{info}

\normalsize
\appendix

\section{The potential of complete binary trees}\label{app:tree-pot}

Let $T$ be a binary tree and $w \colon V(T) \rightarrow \N_+$ be a node weight function. Recall that the \emph{rank} of a node $v \in V(T)$ is $\log w(T_v)$ and the \emph{potential} $\Phi(T)$ of $T$ is defined as the sum of all node ranks. In the following, we explicitly write $\Phi(T,w)$ to denote the potential of $T$ w.r.t.\ the weight function $w$.

The \emph{complete binary tree} $C_n$ on $n = 2^k-1$ nodes is the tree where all leaves are at depth $k$. If $n+1$ is not a power of two, then $C_n$ is constructed from $C_{n'}$ with $n' = 2^{\lceil \log (n+1) \rceil} - 1$ by progressively removing leaves from right to left until the number of nodes is $n$.
In this section, we show:

\begin{theorem}\label{p:tree-pot}
	Let $T = B_n$ and let $w \colon V(T) \rightarrow \N_+$. Let $W = w(T)$ be the sum of weights in $T$. Then $\Phi(T,w) \in \fO(W)$. 
\end{theorem}

The following lemma is central:
\begin{lemma}\label{p:tree-pot-single}
	Let $T = B_n$, let $v \in V(T)$, and let $w, w' \colon V(T) \rightarrow \N_+$ be weight functions such that $w'(u) = w(u)$ for all $u \neq v$ and $w'(v) > w(v)$. Then $\Phi(T,w') \le \Phi(T,w) + 3(w'(v) - w(v))$.
\end{lemma}
\begin{proof}
	Let $x = w'(v) - w(v)$. Let $a_0 = v$ and let $a_1, a_2, \dots, a_k$ be the proper ancestors of $v$ in $T$, ordered from $v$ to the root. The weight increase of $v$ only affects itself and its ancestors, so we have:
	\[ \Phi(T,w') - \Phi(T,w) = \sum_{i=0}^k \log w'(T_{a_i}) - \log w(T_{a_i}) = \sum_{i=0}^k \log \frac{w(T_{a_i})+x}{w(T_{a_i})}. \]
	
	Observe that $w(T_{a_i}) \ge 2^{i+1}-1 \ge 2^i$. Using the fact that $\tfrac{a+x}{a} \le \tfrac{b+x}{b}$ for all $a \ge b$, we get
	\begin{align*}
		\Phi(T,w') - \Phi(T,w) & \le \sum_{i=0}^k \log \frac{2^i+x}{2^i}
			\le \sum_{i=0}^k \log \left( 1 + \frac{x}{2^i} \right)
			\le \sum_{i=0}^k \log e^{x/2^i}\\
			& \le 2 x \log e \le 3x.\tag*{\qedhere}
	\end{align*}
\end{proof}

\begin{proof}[Proof of \cref{p:tree-pot}.]
	Suppose that $n = 2^k-1$; the other cases follow by padding with weight-one leaves.
	
	We start with the case that $w(v) = 1$ for all $v \in V(T)$.
	Then
	\[ \Phi(T,w) = \sum_{i=0}^{k-1} 2^i \log \tfrac{n}{2^i} \in \fO(n), \]	
	as desired.
	
	We can now transform $w$ into an arbitrary weight function $w'$ by repeatedly increasing the weight of some node by some amount $x$. Every time we do this, the potential increases by at most $3x$ by \cref{p:tree-pot-single}. Thus, the total potential increase is at most $3(w'(T)- w(T)) = 3(w'(T)-n)$, for a total potential of $\fO(w'(T) + n) = \fO(w'(T))$.
\end{proof}

\section{The Dijkstra-Jarník-Prim algorithm with working-set heaps}\label{app:jp}

In this section, we give a technical description of how to compute a minimum spanning tree in the time required by \cref{p:univ-opt-et}. The main theorem follows.

\restateJarnikPrim*

We first specify the priority queue data structure used in the Dijkstra-Jarník-Prim algorithm. We use the classical operation names, where our \emph{priorities} are called \emph{keys}.
\begin{definition}
	A \emph{priority queue} maintains a set of elements $S$ and an injective priority assignment $p \colon S \rightarrow \Prios$, under the following operations:
	\begin{itemize}
		\item $\texttt{delete-min}() \rightarrow x$ removes and returns the minimum-priority element $x$ in $S$.
		\item $\texttt{insert}(x, q)$ inserts an element $x$ and sets $p(x) \gets q$.
		\item $\texttt{decrease-key}(x, y)$ sets $p(x) \gets q$, assuming $x \in S$ and $q$ is smaller than the previous value for $p(x)$.
		\item $\texttt{get-key}(x) \rightarrow q$ returns $q = p(x)$.
	\end{itemize}
\end{definition}

\begin{algorithm}
	\caption{The Dijkstra-Jarník-Prim MST algorithm. The input is an edge-weighted graph $(G,w)$ and a starting vertex $s \in V(G)$}\label{alg:jp}
	\begin{algorithmic}
		\Procedure{DJP}{$G, w, s$}
			\State Initialize empty priority queue $Q$
			\State $Q.\texttt{insert}(s, ?)$ \Comment{Priority is arbitrary}
			\State Initialize map $R \colon V(G) \rightarrow V(G) \cup \{\bot\}$ with $\bot$ \Comment{Predecessors}
			\State Initialize map $p \colon V(G) \rightarrow \Prios$ with $\infty$ \Comment{Current priorities}
			\State Initialize set $U \gets \emptyset$ \Comment{Visited vertices}
			\While{$Q$ is not empty}
				\State $u \gets Q.\texttt{delete-min}()$
				\State $U \gets U \cup \{u\}$
				\ForEach{neighbor $v$ of $u$ with $v \notin U$}
					\If{$v \notin Q$}
						\State $Q.\texttt{insert}(v, w(u,v))$.
						\State $R[v] = u$
					\ElsIf{$w(u,v) < p[v]$}
						\State $Q.\texttt{decrease-key}(v, w(u,v))$
						\State $R[v] = u$
					\EndIf
					\State $p[v] = w(u,v)$
				\EndFor
			\EndWhile
			\State \Return Spanning tree $T$ with $E(T) = \{ \{v, R(v)\} \mid v \in V(G), R(v) \neq \bot \}$.
		\EndProcedure
	\end{algorithmic}
\end{algorithm}

Pseudocode for the Dijkstra-Jarník-Prim algorithm is given in \cref{alg:jp}. With a classical priority queue implementation such as \emph{Fibonacci trees}~\cite{FredmanTarjan1987}, with constant-time $\texttt{decrease-key}$, the running time of \cref{alg:jp} is $\fO( |E(G)| + |V(G)| \log |V(G)| )$. We use a stronger priority queue implementation,
which is said to have the \emph{working set property}~\cite{HaeuplerHladikEtAl2025a} or is alternatively called \emph{timestamp-optimal}~\cite{HoogRotenbergEtAl2025}. We use the terminology of van der Hoog, Rotenberg, and Rutschmann here~\cite{HoogRotenbergEtAl2025}.

Consider a sequence of priority queue operations. Define the \emph{timestamp} of each operations to be the number of \texttt{insert} operations executed so far, including the current operation. Observe that multiple operations may share the same timestamp.

\begin{theorem}[{\cite{HaeuplerHlad'ikEtAl2024,HoogRotenbergEtAl2025}}]\label{p:wspq}
	There exists a priority queue implementation, called a \emph{timestamp-optimal priority queue}, where each $\alg{insert}$ and $\alg{decrease-key}$ takes $\fO(1)$ amortized time,
	and each $\alg{delete-min}(x)$ takes $\fO( \log (t - s_x) )$ time, where $t$ is the timestamp of the $\alg{delete-min}$ operation, and $s_x$ is the timestamp of the $\alg{insert}$ operation that inserted $x$.
\end{theorem}

Let $G$ be a graph and $s \in V(G)$. A \emph{linearization} of $(G,s)$ is any ordering of $V(G)$ that can be obtained by assigning edge weights to $G$ and ordering the vertices by distance to $s$ (assuming all distances are distinct)~\cite{HaeuplerHlad'ikEtAl2024,HoogRotenbergEtAl2025}. Let $\Lin(G,s)$ denote the set of linearizations of $(G,s)$. Berendsohn~\cite{Berendsohn2025} showed that linearizations are precisely the orderings produced by the \emph{vertex search antimatroid} of $(G,s)$, which is the set of orderings starting with $s$ where each prefix forms a connected subgraph of $G$. Consequently, the \emph{reverses} of linearizations are obtained by progressively removing vertices other than $s$ without ever producing a disconnected graph. This means that linearizations are in bijection with degenerate elimination trees with root $s$, so $|\Lin(G,s)| \le |\ETs(G)|$.

We now prove that \cref{alg:jp} with a timestamp-optimal priority queue has running time $\fO(m+n+\log|\Lin(G,s)|)$.
By the discussion above and using \cref{p:et-lb-easy}, this implies \cref{p:jp}. We closely follow van der Hoog, Rotenberg, and Rutschmann~\cite{HoogRotenbergEtAl2025}.
Fix an edge-weighted input graph $(G,w)$ with $n$ vertices and $m$ edges, and fix a starting vertex $s \in V(G)$. Run \cref{alg:jp}. Observe that that each vertex is inserted exactly once into $Q$ and deleted exactly once from $Q$. Let $s_x$ denote the timestamp of insertion and let $t_x$ denote the timestamp of deletion. The sequence of deletions forms a linearization, which means we can use the following lemma, originally stated for Dijkstra's algorithm:

\begin{lemma}[{\cite[Lemmas 7, 8]{HoogRotenbergEtAl2025}}]\label{p:rt-lins}
	$\sum_{x \in V(G)} \log(t_x-s_x) \in \fO( n + \log |\Lin(G,s)| )$.

\end{lemma}

By \cref{p:wspq,p:rt-lins}, the total running time of all \texttt{delete-min} operations is $\fO( n + \log |\Lin(G,s)|)$.
Observe that the algorithm performs $\fO(m)$ operations \texttt{insert} and \texttt{decrease-key}, and $\fO(m+n)$ further work, so the running time is $\fO( m + n + \log |\Lin(G,s)|)$.
Since $|\Lin(G,s)| \le |\ETs(G)|$, as observed above, and $|\ETs(G)| \ge \Omega(n)$, by \cref{p:et-lb-easy}, we can conclude \cref{p:jp}.

\end{document}